\documentclass[acmsmall,nonacm]{acmart}

\usepackage{booktabs}
\usepackage{graphicx}
\usepackage{subcaption}
\usepackage{tikz}
\usepackage{xcolor}
\usetikzlibrary{tikzmark}
\usetikzlibrary{calc,fit,positioning}
\usepackage{mdframed}
\usepackage[normalem]{ulem}

\title{On the Expressiveness of Languages for Querying Property Graphs in Relational Databases}

\author{Hadar Rotschield}
\affiliation{
 \institution{School of Computer Science, Hebrew University}
 \city{Jerusalem}
 \country{Israel}
}
\email{hadar.rotschield@mail.huji.ac.il}

\author{Liat Peterfreund}
\affiliation{
 \institution{School of Computer Science, Hebrew University}
 \city{Jerusalem}
 \country{Israel}
}
\email{liat.peterfreund@mail.huji.ac.il}

\usepackage[disable]{todonotes}

\definecolor{green}{RGB}{0,120,0}
\definecolor{liateditcolor}{RGB}{244,210,255}
\definecolor{hadareditcolor}{RGB}{244,210,0}

\newcommand{\OMIT}[1]{}

\newcommand{\pat}{\psi}
\newcommand{\pfunc}{\mathrel{\rightharpoonup}}

\usepackage[final]{listings}
\usepackage{enumitem}

\usepackage{color}
\definecolor{gray}{RGB}{26,30,35}
\definecolor{darkblue}{RGB}{68,89,197}
\definecolor{darkred}{rgb}{0.45,0,0}
\definecolor{darkgreen}{RGB}{0, 156, 0}
\definecolor{darkpurple}{RGB}{120, 0, 180}
\definecolor{ForestGreen}{RGB}{34,139,34}

\renewcommand{\phi}{\varphi}

\newcommand{\df}{:=}
\newcommand{\query}{\mathcal Q}

\newcommand{\db}{\mathcal{D}}
\newcommand{\gdb}{G}

\newcommand{\tup}{\mu}

\newcommand{\join}{\bowtie}

\newcommand{\tupunion}{\join}

\newcommand{\arity}[1]{\mathsf{arity}(#1)}

\newcommand{\dom}[1]{\mathsf{dom}(#1)}

\newcommand{\lbl}{\mathsf{lab}}
\newcommand{\src}{\mathsf{src}}
\newcommand{\tgt}{\mathsf{tgt}}

\newcommand{\prop}{\mathsf{prop}}

\usepackage{tikz}
\usetikzlibrary{calc}
\makeatletter
\let\pgf@orig@path\path
\makeatother
\renewcommand{\path}[1]{\mathsf{path}(#1)}

\newcommand{\concat}{\cdot}

\newcommand{\sch}[1]{\mathsf{fv}\left(#1\right)}
\newcommand{\schb}[1]{\mathsf{fv}\Big(#1\Big)}

\newcommand{\Nodeset}{\mathcal{N}}

\newcommand{\Edgeset}{\mathcal{E}}

\newcommand{\labelset}{\mathcal{L}}
\newcommand{\keyset}{\mathcal{K}}
\newcommand{\constset}{\mathcal C}

\newcommand{\Vars}{\mathsf{Vars}}

\newcommand{\sem}[1]{\left\llbracket#1\right\rrbracket}
\newcommand{\semd}[1]{\left\llbracket#1\right\rrbracket_{\db}}

\newcommand{\op}{\circ}

\theoremstyle{remark}
\newtheorem{remark}{Remark}[section]

\newcommand{\ADOM}{\operatorname{adom}}

\theoremstyle{definition}
\newtheorem{definition}{Definition}[section]
\newtheorem{theorem}{Theorem}[section]
\newtheorem{lemma}[theorem]{Lemma}

\newcommand{\pgView}{\mathsf{pgView}}

\newcommand{\patExtended}[1]{\pat^{\mathsf{ext}}_{#1}}
\newcommand{\pgViewExtendedUnion}{\pgView^{\mathsf{ext}}}

\newcommand{\localsch}{\mathcal S}

\newenvironment{repeatresult}[2]{\vskip0.5em\par\textsc{#1 #2.}\em}{\vskip1em}

\newenvironment{reptheorem}[1]{\begin{repeatresult}{Theorem}{#1}}{\end{repeatresult}}

\newcommand{\return}{\Omega}

\usepackage{textcomp}
\newcommand{\ttpcr}{\renewcommand{\ttdefault}{pcr}\ttfamily}

\newcommand{\sqlkw}[1]{\text{\normalfont\small\ttpcr\bfseries\color{darkblue} #1}}

\lstnewenvironment{sql}[1][]{%
  \lstset{%
    language=SQL,%
    basicstyle=\small\color{gray}\ttpcr,%
    keywordstyle=\bfseries\color{darkblue},%
    morekeywords={REFERENCES,IS,COLUMNS,GRAPH_TABLE,COLUMNS,CREATE, PROPERTY,GRAPH,VERTEX,EDGE,TABLES,PROPERTIES,SOURCE,DESTINATION,LABEL,WITH},%
    deletekeywords={YEAR,ADD},%
  }%
  \lstset{#1}%
}{}
\lstdefinestyle{SQLInline}{
  language=SQL,
  basicstyle=\small\ttfamily,
  keywordstyle=\bfseries\color{darkblue},
}
\DeclareRobustCommand{\sqlinline}[1]{\lstinline[style=SQLInline]!#1!}

\newcommand{\patExtendedN}[1]{\pat^{(n)}_{#1}}
\newcommand{\pgqextn}{\ensuremath{\textsc{PGQ}^{n}}}
\newcommand{\pgqext}[1]{\ensuremath{\textsc{PGQ}^{#1}}}
\newcommand{\fotcn}{\fotc{n}}
\newcommand{\fotc}[1]{\mathrm{FO[TC^{#1}]}}
\newcommand{\tc}[1]{\mathrm{TC^{#1}}}

\lstnewenvironment{gql}[1][]{%
  \lstset{%
    language=SQL,%
    basicstyle=\small\color{gray}\ttpcr,%
    keywordstyle=\bfseries\color{darkblue},%
    morekeywords={REFERENCES,IS,COLUMNS,GRAPH_TABLE,COLUMNS,CREATE, PROPERTY,GRAPH,VERTEX,EDGE,TABLES,PROPERTIES,SOURCE,DESTINATION,MATCH,FILTER,RETURN,WITH},%
    deletekeywords={YEAR},%
  }%
  \lstset{#1}%
}{}

\newcommand{\nlog}{\ensuremath{\textsc{NL}}}
\newcommand{\ptime}{\ensuremath{\textsc{P}}}

\newcommand{\fotcc}{\ensuremath{\textsc{FO[TC]}}}

\newcommand{\currentsidemargin}{%
  \ifodd\value{page}\oddsidemargin\else\evensidemargin\fi
}
\newlength{\whatsleft}
\newcommand{\measureremainder}[1]{%
\begin{tikzpicture}[overlay,remember picture]
  \let\path\pgf@orig@path
  \path (current page.north west) ++(\hoffset, -\voffset)
    node[anchor=north west, shape=rectangle, inner sep=0, minimum width=\paperwidth, minimum height=\paperheight]
    (pagearea) {};
  \path (pagearea.north west) ++(1in+\currentsidemargin,-1in-\topmargin-\headheight-\headsep)
    node[anchor=north west, shape=rectangle, inner sep=0, minimum width=\textwidth, minimum height=\textheight]
    (textarea) {};
  \path let \p0 = (0,0), \p1 = (textarea.east) in
    [/utils/exec={\pgfmathsetlength#1{\x1-\x0}\global#1=#1}];
\end{tikzpicture}%
}

\theoremstyle{definition}

\newcommand{\rwpgq}{\ensuremath{\textsc{PGQ}^{\textsc{rw}}}}
\newcommand{\rpgq}{\ensuremath{\textsc{PGQ}^{\textsc{ro}}}}
\newcommand{\erwpgq}{\ensuremath{\textsc{PGQ}^{\textsc{ext}}}}

\newcommand{\schema}{\mathcal{S}}
\newcommand{\relations}{\mathcal{R}}
\newcommand{\propset}{\mathcal{P}}
\newcommand{\adom}{\mathsf{adom}}

\newcommand{\lab}{\operatorname{lab}}

\newcommand{\barx}{\bar{x}}



\begin{document}

\begin{abstract}

SQL/PGQ is the emerging ISO standard for querying property graphs defined as views over relational data. We formalize its expressive power across three fragments: the read‐only core, the read‐write extension, and an extended variant with richer view definitions. Our results show that graph creation plays a central role in determining the expressiveness.

The read‐only fragment is strictly weaker than the read‐write fragment, and the latter is still below the complexity class $\nlog$. Extending view definitions with arbitrary arity identifiers closes this gap: the extended fragment captures exactly $\nlog$. This yields a strict hierarchy of SQL/PGQ fragments, whose union covers all $\nlog$ queries. On ordered structures the hierarchy collapses: once arity-2 identifiers are allowed, higher arities add no power, mirroring the classical transitive-closure collapse and underscoring the central role of view construction in property graph querying.
\end{abstract}

\maketitle

\lstdefinelanguage{SQL}{
  morekeywords={
   SELECT, CREATE, NODES, TABLE, EDGES, TARGET, LABELS, FROM, WHERE, GRAPH_TABLE, MATCH, SOURCE, KEY, DESTINATION, REFERENCES, LABEL, NODE, EDGE, VERTEX, TABLES, CREAT, PROPERTY, GRAPH, RETURN, JOIN, PROPERTIES, ON, DISTINCT, WITH, AS, UNION, ALL, EXISTS, ANY, RECURSIVE, ORDER, BY, AND
  },
  sensitive=false,
  morecomment=[l]{--},
  morestring=[b]',
}
\lstset{
  basicstyle=\footnotesize\ttfamily,
  keywordstyle=\color{blue},
  commentstyle=\color{gray},
  stringstyle=\color{red},
  numbers=left,
  numberstyle=\tiny,
  stepnumber=1,
  numbersep=10pt,
  breaklines=true,
  showstringspaces=false,
  frame=none
}

\section{Introduction}
\label{sec:intro}
Property graphs are now widely used in industry, supporting key applications in analytics, fraud detection, and recommendation systems~\cite{srivastava2023fraud,mirza2003studying}.
Property graphs have evolved from simpler graph data models such as edge-labeled graphs and RDF triples, offering a richer and more flexible structure~\cite{LDBC:TR:TR-2021-01}. In a property graph, both nodes and edges can be annotated with labels that can be thought of as types, and can carry arbitrary sets of key-value pairs as properties~\cite{surveyChile}. This design allows property graphs to model heterogeneous and complex domains naturally.
Driven by the wide adoption of property graphs in real-world applications, the International Organization for Standardization (ISO), together with partners from academia and industry, is developing a new standard: GQL~\cite{GQLStandards}. The standard has two components: a standalone graph query language GQL, and SQL/PGQ, which integrates graph querying into relational databases. As these standards take shape, there is a growing need for solid theoretical foundations to guide their development and support the design of robust, graph-aware relational engines.

In a similar spirit to~\cite{WhatGoesAround}, we observe that recent developments in graph data management follow a familiar historical trajectory: instead of replacing the relational model, new paradigms like property graphs are being integrated into it. SQL continues to evolve by incorporating core ideas from alternative models, now through extensions like SQL/PGQ, supporting property graph querying natively within the relational framework.

Motivated by this ongoing convergence of graph and relational models, we take a step towards a deeper understanding of SQL/PGQ by formally analyzing its expressive power and identifying the key factors that influence it.

\paragraph*{SQL/PGQ and Property Graph Views}
SQL/PGQ is an extension of SQL for querying property graphs defined as relational views.
The language operates across three conceptual layers: (i) a pattern matching layer, where graph patterns are matched against the view-defined graph resulting in relations; (ii) a relational algebra layer, which manipulates the results of pattern matching using standard SQL constructs; and (iii) a view creation layer, where the actual graph structure, that is, nodes, edges, labels, and properties, is constructed. 
While layers (i) and (ii) have received attention in previous works, the view creation step (iii) has remained under-explored. 

The role of view creation in the read-write fragment of SQL/PGQ is significant.  Intermediate relational results may themselves be used to generate new graphs, enabling recursive-like compositions that entangle all three layers. Our key observation is that the expressive power of SQL/PGQ is ultimately governed by the flexibility given in layer (iii), that is, its graph view definitions. 
Understanding and formally capturing this mechanism is essential for characterizing the expressiveness.

\begin{example}[Graph View: Transfers Between Bank Accounts]\label{ex:bank-transfers}
Consider a relational schema that consists of the relations
\begin{itemize}
    \item \sqlinline{Account(iban)} that lists all bank accounts by their unique IBAN, and 
    \item 
     \sqlinline{Transfer(t\_id, src\_iban, tgt\_iban, ts, amount)} that records transfers identified uniquely by \sqlinline{t\_id} from \sqlinline{src\_iban} to \sqlinline{tgt\_iban} at time \sqlinline{ts}, with amount \sqlinline{amount}.
\end{itemize}
To analyze suspicious sequences of transfers, we can define the following SQL/PGQ graph view:
\begin{center}   
\begin{minipage}{0.8\linewidth}   
\begin{lstlisting}[language=SQL, numbers=none]
CREATE PROPERTY GRAPH Transfers (
  NODES TABLE 
    Account KEY (iban) LABEL Account,
  EDGES TABLE 
    Transfer KEY (t_id)
        SOURCE KEY src_iban REFERENCES Account  
        TARGET KEY tgt_iban REFERENCES Account  
        LABELS Transfer PROPERTIES (ts, amount) );
\end{lstlisting}
\end{minipage}
\end{center}

This creates a (directed) property graph in which
nodes are accounts (identified by \sqlinline{iban}), and edges represent transfers (identified by \sqlinline{t\_id}) from sender to recipient. Edges are labeled and store \sqlinline{ts} and \sqlinline{amount} as properties. On this property graph SQL/PGQ applies pattern matching. 
\end{example}

\paragraph{Problem Statement}
In this work, we take a step toward understanding the expressive power of SQL/PGQ by focusing on the role of graph view definitions. Prior abstractions of SQL/PGQ were proposed~\cite{pods23,icdt23}, with the most recent effort~\cite{vldb25} extracting the core of the read-only fragment but overlooking the graph view construction layer. Our first goal is to provide a precise formalization of the read-only fragment, $\rpgq$, consistent with the standard~\cite{sqlpgq-standard}, and to fill this gap in layer~(iii) by making explicit how graphs are derived from relational data. Building on this formalization, we compare the expressive power of the read-only $\rpgq$ and the read-write $\rwpgq$ fragments. We show that while the read-write variant $\rwpgq$ allows for graph views to be constructed and queried dynamically, enabling more expressive compositions, it still does not capture all queries in $\nlog$, a natural benchmark for expressive power and complexity for graph query languages.

To explain this gap, we identify key features that remain missing from SQL/PGQ's current design, most notably the lack of tuple-based identifiers and richer view constructs. By addressing these limitations, we define an extended read-write fragment, $\erwpgq$, that supports more powerful graph constructions.
This leads to a refined classification of the expressive capabilities of SQL/PGQ as we formally characterize the boundaries between $\rpgq$, $\rwpgq$, and $\erwpgq$.

\paragraph*{Our Contributions.}
This work provides a precise characterization of the expressive power of SQL/PGQ by isolating the role of graph view definitions. Specifically:

\begin{enumerate}
  \item \textbf{Semantics of SQL/PGQ fragments.}  
  We define the formal syntax and semantics of three fragments: the read-only fragment $\rpgq$, the read-write fragment $\rwpgq$, and the extended fragment $\erwpgq$, which supports graph views with arbitrary $n$-ary node and edge identifiers.

  \item \textbf{Capturing $\nlog$ with view-based recursion.}  
  We show that $\erwpgq$ captures exactly the expressive power of $\fotcc$ (first-order logic with transitive closure), and thus all of $\nlog$ on ordered structures. In contrast, $\rwpgq$ remains strictly below this bound, establishing the separation $\rpgq \subsetneq \rwpgq \subsetneq \erwpgq = \fotcc = \nlog$.

 \item \textbf{Arity-based expressiveness.}  
We identify a hierarchy within $\erwpgq$ based on the maximum arity $n$ of node and edge identifiers. Each level $\pgqextn$ corresponds exactly to $\fotcn$, the fragment of first-order logic with $n$-ary transitive closure. The hierarchy is strict up to arity 2, where it collapses (on ordered structures): higher arities add no expressive power. Thus, arity-2 identifiers suffice to capture the full power of $\erwpgq$.

\item \textbf{Descriptive complexity.}  
We establish tight data complexity bounds for all fragments, situating SQL/PGQ’s expressive power within standard complexity classes and connecting it to descriptive complexity theory.

\end{enumerate}

\paragraph*{Technical Preview.}
We link SQL/PGQ to first-order logic with transitive closure (\fotcc) through constructive, bidirectional translations in the spirit of~\cite{FigueiraLinPeterfreund24RelPersp}. This lets us transfer results both ways and isolate what each fragment can express.  Using these translations, we exhibit concrete separator queries: each lies in the stronger fragment but is provably absent from the weaker, pinpointing the impact of view creation and $n$-ary identifiers on recursion.

\paragraph*{Related Work.}
Early foundational work on graph querying centered on \emph{regular path queries} (RPQs) and their conjunctive and two-way extensions~\cite{RPQ,BarceloLLW-tods12,BarceloLR-jacm14}.  
Nested regular expressions further enrich RPQs with controlled nesting and synchronization~\cite{NRE}, and {regular queries} offer a non-recursive Datalog layer and  capture $\nlog$~\cite{reutter2017rq}.  
These formalisms assume \emph{edge-labeled graphs} without node or edge properties, unlike the property graph model  adopted by, e.g., Cypher, G-CORE, and PGQL~\cite{openCypher,gcore,PGQL}.  

Beyond these classical foundations, recent efforts have sought to better align theoretical formalisms with the richer features of practical languages. In particular, dl-CRPQs extend conjunctive regular path queries with node and edge treatment, list variables, path modes, and data filters, thereby mirroring key design choices in Cypher, SQL/PGQ, and GQL~\cite{LibkinMMPV25}. This work situates modern languages within automata-theoretic foundations and highlights opportunities for optimization. 

Our work closes a complementary key gap by isolating two features missing from earlier formal models, view creation and $n$-ary identifiers, and showing that their addition lifts SQL/PGQ to the full power of $\fotcc$. 
This means~\erwpgq~captures exactly the \nlog~queries, positioning it firmly within a well-studied expressiveness class. Importantly, this is achieved without relying on general fixed-point logic, keeping evaluation complexity strictly below \ptime. In contrast to earlier results focused on edge-labeled graphs, our analysis establishes \nlog~expressiveness in the richer settings of property graphs.

\paragraph*{Paper Organization.}  
Section~\ref{sec:prelims} covers background on property graphs and pattern matching. Section~\ref{sec:read-only-pgq} defines the read-only fragment $\rpgq$, and Section~\ref{sec:rwpgq} introduces the read-write fragment $\rwpgq$ and compares its expressiveness to $\rpgq$. Section~\ref{sec:ext-rwpgq} defines the extended fragment $\erwpgq$, and Section~\ref{sec:expressiveness} characterizes its expressive power. 
Section~\ref{subsec:std-gap} discusses deviations of our formalization from the standard, and 
Section~\ref{sec:conclusion} concludes with open research directions. Full proofs and additional details are in the Appendix of the full version of the paper.

\section{Preliminaries}
\label{sec:prelims}

We review the basic definitions of property graphs and relational structures, and describe the pattern matching layer that extracts relations from property graphs.

\subsection{Property Graphs and Relational Structures}

\paragraph*{Property Graphs}
We use the standard definition of property graphs~\cite{pods23}.
Assume pairwise
disjoint countable sets $\labelset$ of labels, $\keyset$ of keys,
$\propset$ of properties, $\Nodeset$ of node IDs, and $\Edgeset$ of
edge IDs.

\begin{definition}[Property Graphs]
	\label{def:pg}
	A property graph  is a tuple 
	$
	\gdb = \langle N, E, \src, \tgt, \lbl, \prop\rangle
	$
	where
	\begin{itemize}
		\item $N \subset \Nodeset$ is a finite set of node IDs used in $\gdb$;
		\item $E \subset \Edgeset$ is a finite set of directed edge IDs used in $\gdb$;
		\item $\src, \tgt: E \to N$ define source and target of an edge;
        	\item $\lbl:  N  \cup  E \to 2^{\labelset}$  
		is a labeling function that associates with every node or edge ID a (possibly empty) finite set
		of labels from $\labelset$;
		\item $\prop: (N  \cup  E) \times \keyset \rightharpoonup \propset$ is a finite partial
		function that associates a property value
        with a node/edge id and a key. 
	\end{itemize}
\end{definition}

\paragraph*{Relational Structures} 

We assume two infinite sets:   
 the set \(\relations\) of relation names $R$, and the set $\constset$ of domain elements.  We follow the unnamed perspective, and associate 
each relation name \(R \in \relations\) with a positive integer   $\arity{R}$.
A \emph{(database) schema} \(\schema \subseteq \relations\) is a finite set of relation names.  
A \emph{relation} $\mathbf{R}$ over a relation name $R$ is a finite set of \emph{tuples} in $\constset^{\arity{R}}$.
A \emph{database (instance)} 
$\db$ over a schema 
$\schema$
assigns to each relation name 
$R\in \schema$ a 
\emph{relation}
$R^{\db}$ over $R$.
The active domain of $\adom(\db)$ of $\db$ consists of all constants that appear in $\db$.
\begin{remark}
    Throughout the paper, we assume our structures are ordered, reflecting systems where order is implicit and aligning with standard practice in database theory.
\end{remark}

\paragraph{Relations as (partial) functions}
A relation \( R_f \) is said to \emph{encode a function} \( f : X \to Y \) if
$
\arity{R_f} =  \arity{X}+\arity{Y},
$
and for every \(x \in X\), there exists a unique \(y \in Y\) such that \((x, y) \in R_f\).
Similarly, \( R_f \) encodes a \emph{partial function} \( f : X \rightharpoonup Y \) if the same condition on the arities holds, and for every \(x \in X\), there exists at most one \(y \in Y\) such that \((x, y) \in R_f\).

\subsection{Pattern Matching: From Property Graphs to Relations}
Pattern matching is the key component of graph query languages as it extracts relations from property graphs.
We use the refined definitions of core PGQ from~\cite{vldb25}.
We fix an infinite set $\Vars$ of
variables and define \emph{patterns} $\pat$ and \emph{output patterns} $\pat_\return$
 in Figure~\ref{fig:pattern-syntax}.

\begin{figure*}[t]
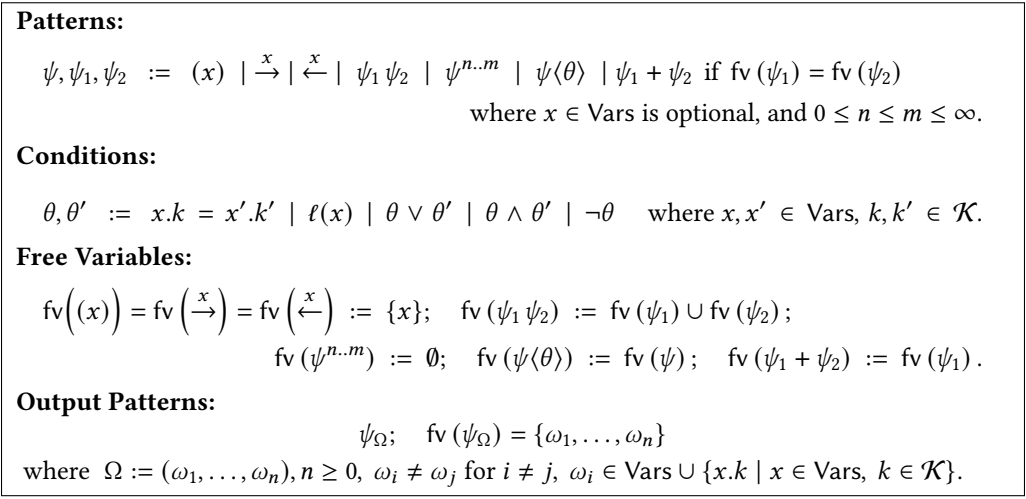

\centering
\fbox{
\begin{minipage}{0.95\textwidth}
\textbf{Patterns:}
\begin{multline*}
    	\pat, \pat_1, \pat_2 \ \ \df \ \  (x) \ \mid \
	\overset{x}{\rightarrow} \ \mid \
	\overset{x}{\leftarrow} \ \mid \
	\pat_1\, \pat_2  \ \mid\  
	\pat^{n..m} \ \mid\ 
	\pat\langle\theta \rangle \ \mid 
	\pat_1 + \pat_2 \,\text{ if }\, \sch{\pat_1}=\sch{\pat_2}\\
\text{where } x \in \Vars \text{ is optional, and }  0 \leq n \leq m \leq \infty.
\end{multline*}
\smallskip
\textbf{Conditions:}
\begin{multline*}
    \theta, \theta' \ \df \ x.k=x'.k' \mid
\ell(x) \mid 
\theta \vee \theta'
\mid
\theta \wedge \theta'
\mid \neg \theta \quad  \text{ where }
 x,x'\in\Vars, \, k,k'\in\keyset.
\end{multline*}

\textbf{Free Variables:}
\begin{multline*}
\schb{(x)} = \sch{\overset{x}{\rightarrow}} =  \sch{\overset{x}{\leftarrow}} \ \df \ \{x\}; \quad 
\sch{ \pat_1\, \pat_2 } \ \df\ \sch{\pat_1}\cup\sch{ \pat_2 }; \\ 
\sch{   \pat^{n..m} } \ \df\ \emptyset; \quad
\sch{\pat\langle\theta \rangle}  \ \df\ \sch{\pat};\quad
 \sch{\pat_1 + \pat_2} \ \df\  \sch{\pat_1}.
\end{multline*}

\textbf{Output Patterns:}
\[
\pat_\return; \quad \sch{\pat_\return}=\{ \omega_1,\ldots, \omega_n\}
\]
\text{ where }
$\return \df (\omega_1, \ldots, \omega_n), 
n \geq 0,\ \omega_i \neq \omega_j\ \text{for } i \neq j,\ \omega_i \in \Vars \cup \{x.k \mid x \in \Vars,\ k \in \keyset\}.
$

\end{minipage}
}
\caption{Syntax of Patterns and Output Patterns.}
\Description{Syntax of patterns and their free variables in SQL/PGQ.}
\label{fig:pattern-syntax}
\end{figure*}

\subsection{Pattern Matching Semantics}
\label{sec:pattern-semantics}
In Figure~\ref{fig:sem-endpoints}, we present  core PGQ pattern matching semantics as originally introduced in~\cite{vldb25}, and later further simplified in~\cite{FigueiraLinPeterfreund24RelPersp}.

\subsubsection{Patterns}
The semantics $\sem{\pat}_G$ of a pattern $\pat$ on $G$ is defined as a set of triples \((s, t, \mu)\), where \(s\) and \(t\) are the source and target nodes of a path matching $\pat$, and \(\mu\) is a \emph{variable mapping} assigning matched graph elements (that is, nodes and edges) to the pattern’s free variables.\footnote{Our simplification is based on the fact that the semantics avoids storing the full path, and instead it preserves the key information needed for composing patterns, that is, its source and target.} 

The semantics uses the following definitions:
\paragraph{Operations on variable mappings.}
We use standard operations on mappings in the semantics.  
Given a mapping \(\mu\) and a subset \(X\subseteq \dom{\mu}\), the restriction \(\mu \restriction X\) is the mapping with domain \(X\) such that \((\mu \restriction X)(x) \df \mu(x)\) for all \(x \in X\).
We write $\mu_{\emptyset}$ to denote the mapping with an empty domain.

Two mappings \(\mu_1\) and \(\mu_2\) are \emph{compatible}, written \(\mu_1 \sim \mu_2\), if they agree on all common variables. Their \emph{union}, written \(\mu_1 \tupunion \mu_2\), is the mapping defined over \(\dom{\mu_1} \cup \dom{\mu_2}\), with values taken from \(\mu_1\) and \(\mu_2\) accordingly.

\paragraph{Condition Satisfaction.}
A mapping \(\mu\) satisfies a condition \(\theta\), written \(\mu \models \theta\), as follows:
\begin{itemize}
  \item \(\mu \models x.k = x'.k'\) if both \(\prop(\mu(x), k)\) and \(\prop(\mu(x'), k')\) are defined and equal;
  \item \(\mu \models \ell(x)\) if \(\ell \in \lbl(\mu(x))\).
\end{itemize}
Satisfaction extends to Boolean combinations of conditions with the standard semantics of \(\wedge\), \(\vee\), \(\neg\).

\subsubsection{Output Patterns}

The semantics $\sem{\pat_{\return}}_G$ of output patterns $\pat_{\return}$ on $G$ is a set of tuples:
\[
(\mu_{\return}(\omega_1), \ldots, \mu_{\return}(\omega_n)),
\]
where each \(\mu_{\return}(\omega_i)\) is either a node identifier, an edge identifier, or a property value. This set is finite because 
\(\mu_{\return}\)
ranges over mappings with finite domain and finite codomain.

To ensure compatibility with the relational model introduced in Section~\ref{sec:prelims}, we assume without loss of generality that
$
\Nodeset \cup \Edgeset \cup \propset \subseteq \constset.
$
That is, all output values belong to the relational domain \(\constset\), so that the result of evaluating an output pattern can be interpreted as a relation. This aligns with the relational layer of PGQ, which allows further use of relational algebra on the output of pattern matching.

\begin{example}[Pattern Matching on Transfers Graph]\label{ex:pattern-transfer-interval}
Consider the property graph defined in Example~\ref{ex:bank-transfers}. Suppose we want to find pairs of account IBANs \((x,y)\) such that there exists a non-empty sequence of transfers from \(x\) to \(y\) where each transfer in this sequence has an amount greater than $100$. 
To this end, we use the following output pattern:
\[
   \left( (x) \overset{t}{\rightarrow}^{1,\infty}(y) \langle \theta \rangle \right)_{x.\texttt{iban},y.\texttt{iban}} \quad  \text{where} \quad 
    \theta \df \texttt{Transfer}(t) \land t.\texttt{amount}>100
\]
which is written in SQL/PGQ as follows:
\begin{center}   
\begin{minipage}{0.8\linewidth}
\begin{lstlisting}[language=SQL, numbers=none]
SELECT * 
FROM GRAPH_TABLE (Transfers MATCH (x) -[t:Transfer]-> +(y) 
                            WHERE t.amount > 100
                            RETURN (x.iban,y.iban) );
\end{lstlisting}
\end{minipage}
\end{center}
\end{example}
\begin{figure*}[t]
\centering
\setlength{\tabcolsep}{0pt}
\fbox{
$
\begin{array}{rl}
\multicolumn{2}{l}{\textbf{Patterns:}}\\[2pt]
\sem{(x)}_G
  &\!:=\!
    \bigl\{\, (n,n,\{x\mapsto n\}) \mid n\in N \bigr\}
\\
\sem{\overset{x}{\rightarrow}}_G
  &\!:=\!
    \bigl\{\, (n_1,n_2,\{x\mapsto e\})
          \mid e\in E,\; \src(e)=n_1,\; \tgt(e)=n_2 \bigr\}
\\
\sem{\overset{x}{\leftarrow}}_G
  &\!:=\!
    \bigl\{\, (n_2,n_1,\{x\mapsto e\})
          \mid e\in E,\; \src(e)=n_1,\; \tgt(e)=n_2 \bigr\}
\\
\sem{\pat_1 + \pat_2}_G
  &\!:=\!
    \sem{\pat_1}_G \cup \sem{\pat_2}_G
\\
\sem{\pat_1\,\pat_2}_G
  &\!:=\!
    \bigl\{\, (s_1,t_2,\mu_1 \join \mu_2)
          \,\big|\,
          (s_1,n,\mu_1)\in\sem{\pat_1}_G,\;
          (n,t_2,\mu_2)\in\sem{\pat_2}_G,\mu_1 \sim \mu_2\bigr\}
\\
\sem{\pat\langle\theta\rangle}_G
  &\!:=\!
    \bigl\{\, (s,t,\mu)\in\sem{\pat}_G \mid \mu\models\theta \bigr\}
\\
\sem{\pat^{n..m}}_G
  &\!:=\!
    \displaystyle\bigcup_{i=n}^{m} \sem{\pat}^{\,i}_G \quad {\text{where}\quad
  \sem{\pat}^{\,0}_G := \bigl\{\, (n,n,\mu_\emptyset) \mid n\in N \bigr\}, \quad \text{and for } n>0}
\\
  \sem{\pat}^{\,n}_G   &\!:=\! \bigl\{
      (s_{1},t_{n},\mu_\emptyset)
      \ \Bigm|\ 
      \exists \mu_{1},\ldots,\mu_{n}:\;
             (s_{i},t_{i},\mu_{i}) \in \sem{\pat}_G
             \;\text{and}\;
             t_{i}=s_{i+1}\;\text{for all } i < n
    \bigr\},
\\[2pt]
\multicolumn{2}{l}{\textbf{Output Patterns:}}\\[2pt]
\sem{\pat_\return}_G
  &\!:=\!
    \bigl\{\,
(\mu_\return(\omega_1),\ldots , \mu_\return(\omega_n))
    \mid \exists s,t:\,(s,t,\mu)\in\sem{\pat}_G, \return \df (\omega_1,\ldots,\omega_n) \bigr\}
\end{array}
$}
\caption{Semantics of Patterns and Output Patterns.} 
\Description{Semantics of graph patterns using source and target node semantics.}
\label{fig:sem-endpoints}
\end{figure*}
\subsection{Expressiveness and Evaluation Complexity}
\label{sec:query-languages-expressivness}

\paragraph*{Query Evaluation Problem}
For a query language~$\mathcal{L}$, the \emph{evaluation problem} asks:  
given a  query $\query \in \mathcal{L}$, a database $\db$, and a candidate tuple~$\bar{a}$,  
decide whether $\bar{a}$ is in the result of evaluating $\query$ on $\db$. 
We focus on the \emph{data complexity} of this problem, that is, the query $\query$ is fixed,  
and the complexity is measured with respect to the size of the input database~$\db$.

\paragraph*{Complexity Class~\nlog}
We use \nlog~to denote the class of decision problems that can be solved
by a nondeterministic Turing machine using $O(\log n)$ work
tape, where $n$ is the size of the input.

\paragraph*{Language Expressiveness}
We say that a query language $\mathcal{L}_2$ is \emph{at least as expressive as} $\mathcal{L}_1$ if for every query $\query \in \mathcal{L}_1$ there exists a query $\query' \in \mathcal{L}_2$ such that $\semd{\query} = \semd{\query'}$ on all databases $\db$.
If, in addition, there exists a query in $\mathcal{L}_2$ not expressible in $\mathcal{L}_1$, then $\mathcal{L}_2$ is \emph{strictly more expressive}, written $\mathcal{L}_1 \subsetneq \mathcal{L}_2$.

\section{Read-Only PGQ}
\label{sec:ropgq}
\label{sec:read-only-pgq}
Building on the formalization of core PGQ presented in~\cite{vldb25}, we further develop the read-only fragment of the language and revisit certain aspects that were only briefly addressed in that work. In particular, we focus on the foundational question: what property graph does the query operate on?
While this issue was understandably not a central focus of GQL, which operates directly on property graphs, a more precise formalization is essential for PGQ, which is defined over relational databases and constructs property graph views dynamically.

\subsection{Property Graph Views}
\label{sec:pgView}
{
We now describe the mechanism for constructing a \emph{property graph view}, or \emph{tabular property graph}, following the Standard’s terminology, from a relational database. This view is only defined when the input database satisfies specific structural constraints, formalized below.

\begin{definition}
    A \emph{property graph view} is a relational database $\db$ over a schema consisting of 
 six relation symbols, unary $R_1,R_2$, binary  $R_3,R_4,R_5$, and a ternary $R_6$, for which the following hold:
 \begin{enumerate}
  \item\label{valid:disjoint}
        \(R_1^{\db}\) and \(R_2^{\db}\) are disjoint relations 
        (representing node and edge identifiers, respectively).
        
  \item\label{valid:src-tgt}
        \(R_3^{\db}, R_4^{\db}\) 
        encode functions  
        \(R^{\db}_2 \rightarrow R^{\db}_1\) 
        (representing the source and target of each edge, respectively).

  \item\label{valid:lab}
        \(R^{\db}_{5} \subseteq (R^{\db}_{1} \cup R^{\db}_{2}) \times \constset\) 
        (representing the labeling of nodes and edges).

  \item\label{valid:prop}
        \(R^{\db}_{6}\) encodes a partial function 
        \((R^{\db}_{1} \cup R^{\db}_{2}) \times \constset \rightharpoonup \constset\) 
        (representing properties).
 \end{enumerate}
\end{definition}
Intuitively, $R_1^{\db}$ and $R_2^{\db}$ identify nodes and edges, 
 $R_3^{\db}$ and $R_4^{\db}$ assign each edge a source and target node, and 
 $R_5^{\db}$ and $R_6^{\db}$ 
capture labels and key-value properties of nodes and edges, respectively.
Note that $R_5^{\db}$ and $R_6^{\db}$ 
may be empty, reflecting a scenario in which there are no labels or properties.
The previous definition specifies when a relational database can be interpreted as a property graph. We now define the partial function 
$\pgView$, which formalizes this interpretation.

\begin{definition}\label{def:pgView}
\sloppy{
We define \(\pgView\) as the partial function that,  
given a sequence of relations $(\mathbf{R}_1, \ldots, \mathbf{R}_6)$ that satisfies the property graph view conditions,  
returns the property graph \(\langle N, E, \src, \tgt, \lbl, \prop \rangle\), denoted \(\pgView(\mathbf{R}_1, \ldots, \mathbf{R}_6)\), where:}
\[
N \df \mathbf{R}_1, \quad
E \df \mathbf{R}_2, \quad
\src \df \mathbf{R}_3, \quad
\tgt \df \mathbf{R}_4, \quad
\lbl \df \mathbf{R}_5, \quad
\prop \df \mathbf{R}_6.
\]
\end{definition}

In PGQ, one applies pattern matching  on such property graph views.

\subsection{Read-Only PGQ: The Language}
\label{sec:read-only-core-pgq}
Output patterns return {relations (finite sets of tuples)}, which in PGQ can be further manipulated using standard relational algebra.
Since pattern matching operates over property graphs, we must interpret certain parts of the input relational database as defining such a graph. To formalize this, we define \emph{read-only PGQ queries} namely $\rpgq$. A query $\query\in \rpgq$ over a relational schema \(\localsch\) is defined in Figure~\ref{fig:grammar-core-pgq}, in the \rpgq~part.

\begin{figure}[ht]
  \centering
  \[
  \fbox{$\displaystyle
    \begin{array}{rl}
      \multicolumn{2}{l}{\textbf{\rpgq:}}\\[2pt]
      \query & \df\;
        \pat_{\return}(\bar R)
        \mid R
        \mid \pi_{\$i_1,\ldots,\$i_k}(\query)
        \mid \sigma_{\theta}(\query)
        \mid \query \times \query'
        \mid \query \cup \query'
        \mid \query - \query' \\[2pt]
      \theta & \df\;
        \$i_1 = \$i_2
        \mid \neg \theta
        \mid \theta \vee \theta
        \mid \theta \wedge \theta\,\,\,\,{\text{where } i_1,\ldots,i_k \in \mathbb{N}}\\[4pt]

      \multicolumn{2}{l}{\textbf{\rwpgq:}}\\[2pt]
      \query & \df\;  \ c \,\mid \pat_{\return}(\bar\query) \\[4pt]

      \multicolumn{2}{l}{\textbf{\erwpgq:}}\\[2pt]
      \query & \df\; \patExtended{\return}(\bar\query)
    \end{array}
  $}
  \]
\caption{\label{fig:grammar-core-pgq} Syntax of $\rpgq$, $\rwpgq$, and $\erwpgq$: 
Each part shows the additions to the previous fragment: $\rwpgq$ extends $\rpgq$ with constants and pattern matching on query results, and $\erwpgq$ further extends $\rwpgq$ with pattern matching on more elaborate query results.
}
\Description{\label{fig:grammar-core-pgq} Syntax of $\rpgq$, $\rwpgq$, and $\erwpgq$}
\end{figure}

Here, \(\bar R = (R_1, \ldots, R_6)\) is a tuple of relation names from \(\localsch\), 
defining the \emph{property graph view} on which the output pattern \(\pat_{\return}\) is evaluated.
The semantics of the relational algebra operators follow their standard definitions; the novel aspect lies in the evaluation of the output pattern on the graph view. We define the semantics of a query \(\query\) over a database \(\db\) as a relation \(\semd{\query}\), given in Figure~\ref{fig:core-pgq-semantics} in the \rpgq~part.

\begin{figure}[ht]
  \centering
  \begin{tikzpicture}[remember picture, baseline=(current bounding box.north)]
    \node[anchor=north west, inner sep=0pt] (mathblock) at (0,0) {
      \begin{minipage}{0.95\textwidth}
        \setlength{\abovedisplayskip}{0pt}
        \setlength{\abovedisplayshortskip}{0pt}
        \setlength{\belowdisplayskip}{0pt}
        \setlength{\belowdisplayshortskip}{0pt}

        \begin{align*}
          \intertext{\textbf{\rpgq semantics:}}
          \semd{\pat_{\return}(\bar R)}
            &\df \sem{\pat_{\return}}_G,
              \quad \text{where } G = \pgView( \semd{R_1},\ldots,\semd{R_6}) \\[2pt]
          \semd{R} &\df R^{\db} \\[2pt]
          \semd{\pi_{\$i_1,\ldots,\$i_k}(\query)}
            &\df \{ (t_{i_1},\ldots,t_{i_k}) \mid (t_1,\ldots,t_n) \in \semd{\query},\ 1\le i_1,\ldots,i_k\le n \} \\[2pt]
          \semd{\sigma_{\theta}(\query)}
            &\df \{ \tup \mid \tup \in \semd{\query},\  \tup \models \theta \} \\[2pt]
          \semd{\query \times \query'}
            &\df \{ (\bar t , \bar t' ) \mid \bar t \in \semd{\query},\ \bar t' \in \semd{\query'} \} \\[2pt]
          \semd{\query \op \query'}
            &\df \semd{\query} \op \semd{\query'}, \quad \text{for } \op \in \{ \cup, \setminus \} \\[4pt]
          \intertext{\textbf{\rwpgq semantics:}}
          \tikzmark{start}
          \semd{c} &\df c \quad \text{where } c \in \adom(\db) \\[2pt]
          \semd{\pat_{\return}(\query_1,\ldots,\query_6)}
            &\df \sem{\pat_{\return}}_{G}
              \quad \text{where }
              G \df \pgView\!\left(\semd{\query_1},\dots,\semd{\query_6}\right)
          \tikzmark{end}
          \intertext{\textbf{\erwpgq semantics:}}
          \tikzmark{start}
          \semd{\patExtended{\return}(\query_1,\ldots,\query_6)}
            &\df \sem{\pat_{\return}}_{G}
              \quad \text{where }
              G \df \pgViewExtendedUnion\!\left(\semd{\query_1},\dots,\semd{\query_6}\right)
          \tikzmark{end}
        \end{align*}
      \end{minipage}
    };
    \begin{scope}[overlay, remember picture]
      \draw[thick]
        ([xshift=-8pt,yshift=-6pt]mathblock.south west) rectangle
        ([xshift=8pt,yshift=6pt]mathblock.north east);
    \end{scope}
  \end{tikzpicture}
  \caption{Semantics of core PGQ queries.}
  \Description{Semantics of core PGQ queries.}
  \label{fig:core-pgq-semantics}
\end{figure}

\noindent
Given a tuple \(\bar{t} = (t_1, \ldots, t_n)\), we write \(\bar{t} \models \theta\) to indicate that \(\bar{t}\) satisfies the condition \(\theta\), with the following standard semantics:
\[
\bar{t} \models \$i = \$i' \quad \text{iff} \quad 1 \le i, i' \le n \ \text{and} \ t_i = t_{i'}.
\]
This is extended to Boolean connectives \(\wedge, \vee, \neg\) in the usual way.
The function 
$\pgView$ is defined as in Definition~\ref{def:pgView}.

The semantics of $\pat_{\return}(R_1,\ldots,R_6)$ on a database $\db$ are 
\[
   \semd{\pat_{\return}(R_1,\ldots,R_6)}
\, \df \,
\sem{\pat_{\return}}_{G}
\quad \text{
where} \quad G \df   \pgView\left(
      \semd{R_{1}},\dots,\semd{R_{6}}
   \right).\]
That is, to evaluate the query $\pat_{\return}(\bar R)$, we first evaluate each of the six relational subqueries $R_1$ through $R_6$ on the current database instance $\db$. These results define the graph $G$ via the $\pgView$ operator. Then we evaluate the pattern $\pat_{\return}$ over $G$ to produce the result.

\section{Read-Write PGQ}
\label{sec:rwpgq}

In the read-only fragment, pattern matching is restricted to relations that are explicitly stored in the database. In contrast, the read-write fragment extends this by allowing pattern matching to be applied to property graph views constructed dynamically from queries. These queries can be defined recursively, enabling the application of pattern matching not only on the original relations, but also on the results of previous pattern matchings, relational algebra operations, or combinations thereof. This recursive composition supports a powerful form of query nesting and reuse, central to the expressive power of the read-write fragment.

The syntax and semantics additions for \emph{read-write PGQ queries} can be found in \rwpgq~sections in Figures~\ref{fig:grammar-core-pgq} and~\ref{fig:core-pgq-semantics}, respectively.

In particular, the semantics of a query ${\pat_{\return}(\query_1,\ldots,\query_6)}$  on a database $\db$ is 
\[
   \semd{\pat_{\return}(\query_1,\ldots,\query_6)}
\, \df \,
\sem{\pat_{\return}}_{G}
\quad \text{
where} \quad G \df   \pgView\left(
      \semd{\query_{1}},\dots,\semd{\query_{6}}
   \right).\]

That is, we first evaluate each of the six relational subqueries \(\query_1\) through \(\query_6\) on the current database instance \(\db\). These results define the graph \(G\) via the \(\pgView\) operator, and the pattern \(\pat_{\return}\) is then evaluated over \(G\).

{This generalizes the simpler case} 
\(\pat_{\return}(\bar R)\) of \rpgq, where each \(R_i\) is a base relation rather than a query.

We now turn to compare the expressive power of the read-only and read-write fragments.

\subsection{\texorpdfstring{Expressiveness of $\rwpgq$~Vs.~$\rpgq$~}{Expressiveness of rwPGQ Vs. rPGQ}}
Core PGQ~\cite{vldb25} is effectively relational algebra over pattern matching outputs. 
Since relational algebra is equivalent to relational calculus, it follows that \(\mathrm{FO}\subseteq\rpgq\).
Quite expectedly, read-write PGQ queries add expressiveness to the read-only fragment.

\newcommand{\thmrwsupro}{$\rwpgq$~is strictly more expressive than $\rpgq$.}
\begin{theorem}\label{thm:rwsupro}
  \thmrwsupro
\end{theorem}

The proof relies on the following  separating example.
Consider a graph whose vertices are partitioned into \texttt{RedNodes} and \texttt{BlueNodes}, and whose \texttt{Edges} always connect nodes of opposite colors.  
Detecting a path of \emph{arbitrary} length whose colors alternate red-blue at every step separates the two fragments.  
In this example, any $\rpgq$ query can be rewritten into one that does not apply pattern matching, and therefore in relational algebra. Using Gaifman locality theorem~\cite{GAIFMAN1982105}, this query can inspect paths only up to a fixed radius determined by its size, so it cannot recognize all unbounded alternating paths.  
In contrast, $\rwpgq$ can first materialize a view whose node set is \texttt{RedNodes}\,$\cup$\,\texttt{BlueNodes} and whose edge set is \texttt{Edges}, and then apply the reachability pattern \[\left(((x) \rightarrow (y) \rightarrow (z) {\langle \texttt{RedNodes}(x) \wedge \texttt{BlueNodes}(y) \wedge \texttt{RedNodes}(z) \rangle)^*}\right)_{\emptyset}.\]
This query returns true exactly when an alternating-color path with at least two edges exists, which concludes the proof.   

Although $\rwpgq$ is more expressive than its read-only counterpart, its Boolean queries still do not capture all the problems in $\nlog$. This complexity class, $\nlog$, is a widely accepted benchmark for graph query languages because it balances expressive power and efficiency. It includes many path-based queries and corresponds to Datalog's capabilities on CRPQs~\cite{C2RPQ}, as well as SQL’s \sqlkw{WITH RECURSIVE}, which supports linear recursion.

\newcommand{\thmrwsubsetnl}{
$\nlog$ is strictly more expressive than  Boolean
$\rwpgq$.}
\begin{theorem}\label{thm:rw_subset_nl}
   \thmrwsubsetnl
\end{theorem}

This highlights an expressiveness gap between
$\rwpgq$ and the full range of tractable graph queries.  Bridging this
gap motivates the extended fragment $\erwpgq$ studied next.

\section{Extensions of Read-Write PGQ}
\label{sec:extensions}
\label{sec:ext-rwpgq}
{

In the \emph{extended read-write fragment} \erwpgq, whose concrete syntax and formal
semantics are set out in Figures~\ref{fig:grammar-core-pgq} and~\ref{fig:core-pgq-semantics}, we lift the restriction that node and edge identifiers are
single values: they may now be $n$-ary tuples for any fixed $n \ge 1$.
With this single generalization, each clause of Definition~\ref{def:pgView} carries over (1)-(4) with the relation arities scaled accordingly.

\begin{definition}
\label{def:pgView-n}
 An  \emph{$n$-ary property graph view} is a relational database \(\db\) over a schema consisting of six relation
symbols-{$n$-ary $R_1,R_2$, $2n$-ary $R_3,R_4$, an
$(n\!+\!1)$-ary $R_5$, and an $(n\!+\!2)$-ary} $R_6$-for which conditions (1)-(4) from Definition~\ref{def:pgView} hold. That is, 
\begin{enumerate}
  \item\label{valid:disjoint-n}
        $R_1^{\db}$ and $R_2^{\db}$ are disjoint relations
        (representing node and edge identifiers, respectively).

  \item\label{valid:src-tgt-n}
        $R_3^{\db}, R_4^{\db}$ encode functions
        $R_2^{\db} \rightarrow R_1^{\db}$
        (representing the source and target of each edge, respectively).

  \item\label{valid:lab-n}
        $R_5^{\db} \subseteq (R_1^{\db} \cup R_2^{\db}) \times
        \constset$
        (representing the labeling of nodes and edges).

  \item\label{valid:prop-n}
        $R_6^{\db}$ encodes a partial function
        $(R_1^{\db} \cup R_2^{\db}) \times \constset \rightharpoonup
        \constset$
        (representing properties).
\end{enumerate}
\end{definition}
Indeed, this is a generalization of Definition~\ref{def:pgView}, since for $n = 1$ the two definitions coincide.

\begin{remark}\label{rmk:arity-remark}
In our definition of 
$n$-ary property graphs, node and edge identifiers share the same arity to simplify the model: a single pair of relations $R_5,R_6$
 suffices for labels and properties. This choice reduces complications without loss of generality. Allowing different arities for nodes and edges requires duplicating these relations, but all definitions and results extend naturally to that case.
\end{remark}

\begin{definition}
\label{def:pgView-eq-n-func}
For every integer $n\ge 1$, {we define \(\pgView^{=n}\)} as the partial function
that, 
given a sequence of relations $(\mathbf{R}_1, \ldots, \mathbf{R}_6)$ that satisfy the $n$-ary property graph view  (Definition~\ref{def:pgView-n}), returns the corresponding property graph \(\langle N, E, \src, \tgt, \lbl, \prop \rangle\), denoted by \(\pgView^{=n}(\mathbf{R}_1, \ldots, \mathbf{R}_6)\) where 
\[
N \df \mathbf{R}_1, \quad
E \df \mathbf{R}_2, \quad
\src \df \mathbf{R}_3, \quad
\tgt \df \mathbf{R}_4, \quad
\lbl \df \mathbf{R}_5, \quad
\prop \df \mathbf{R}_6.
\]
\end{definition}

\begin{definition}
\label{def:pgView-n-func}
{
For every integer $n \ge 1$, we set $\pgView^{n}$ to be the union of $\pgView^{=i}$ up to $i=n$, that is, 
\[
  \pgView^{n} \;:=\; \bigcup_{i=1}^{n}\pgView^{=i},
\]}
and $\pgViewExtendedUnion$ as the union of all $i$'s:
\[
  \pgViewExtendedUnion \; \df \; \bigcup_{i \ge 1} \pgView^{=i}.
\]
That is, \(\pgViewExtendedUnion\) is the partial function whose domain is the
union of the domains of all \(\pgView^n\).  Concretely, for every sequence of
relations \((\mathbf{R}_1,\ldots,\mathbf{R}_6)\) that satisfies the
\(k\)-ary property graph view conditions for some \(k \ge 1\), we set
\[
  \pgViewExtendedUnion(\mathbf{R}_1,\ldots,\mathbf{R}_6)
  \;\df\;
  \pgView^{k}(\mathbf{R}_1,\ldots,\mathbf{R}_6).
\]
\end{definition}

Notice that $\pgView^{1}=\pgView$, hence $\pgView^k$ extends $\pgView$. 
Note also that \(\pgViewExtendedUnion\) is well defined.

}

We define the extended read-write PGQ, namely $\erwpgq$ by adding $\patExtended{\return}(\bar\query)$ to the syntax grammar. A query $\query \in \erwpgq$ is defined by the additional syntax in the $\erwpgq$ section in Figure~\ref{fig:grammar-core-pgq}:
\[
   \query \ \df \patExtended{\return}(\query_1,\ldots,\query_6).
\]
The semantics is defined using $\pgViewExtendedUnion$ as follows:
\[
  \semd{\patExtended{\return}(\query_1,\ldots,\query_6)}
    := \sem{\pat_\return}_G,
  \quad\text{where }G = \pgViewExtendedUnion
        \bigl(
          \semd{\query_1},\ldots,
          \semd{\query_6}
        \bigr),
\]
and can be found in $\erwpgq$ section in Figure~\ref{fig:core-pgq-semantics}.
\(\erwpgq\) preserves the two-phase evaluation procedure of \(\rwpgq\) but replaces the operator \(\pgView\) with its composite-identifier variant \(\pgViewExtendedUnion\), and substitutes the core pattern \(\pat_\return\) with the extended pattern \(\patExtended{\return}\), whose semantics supports identifiers of arity~\(k\).
In particular, in Figure~\ref{fig:sem-endpoints}, each output triple $(s,t,\mu)$ has $s$ and $t$ as $k$-tuples (rather than arity-1), valuations $\mu$ map variables to $k$-tuples, and all equalities in the $n$-fold repetition are over $k$-tuples.

\paragraph{Expressiveness.}
Since every \(\rwpgq\) query is also an \(\erwpgq\) query, we have the
immediate containment \(\rwpgq \subseteq \erwpgq\).
The strictness of this containment is not evident from the syntax alone,
but will follow from the capture result proved in the next section
and from Theorem~\ref{thm:rw_subset_nl}.

\begin{example}[Extended Graph View with Composite Identifiers]\label{ex:bank-account-ids}
Building on Examples~\ref{ex:bank-transfers} and~\ref{ex:pattern-transfer-interval}, assume now that accounts in our database are no longer identified by a single IBAN column. Instead, an account is uniquely identified by the triple (\sqlinline{bank}, \sqlinline{branch}, \sqlinline{acct}). Thus, the relational schema is adapted, replacing the IBAN field by three separate columns:
\begin{itemize}[leftmargin=*,align=parleft]
\item \sqlinline{Account(bank, branch, acct)} identifies accounts by triples.
\item \sqlinline{Transfer(t\_id, bankSrc, branchSrc, acctSrc, bankTgt,  branchTgt, acctTgt, ts, amount)} 

stores transfers between accounts identified by triples.
\end{itemize}
To capture this structure, we define the following property graph view with composite $n$-ary identifiers for nodes and edges:

$R_1$ is the table \sqlinline{Account(bank,branch,acct)};
$R_2$ is the projection of \sqlinline{Transfer} onto \sqlinline{t\_id};
$R_3$ links every edge in $R_2$ to its source account \sqlinline{(bankSrc, branchSrc,acctSrc)};
$R_4$ links every edge in $R_2$ to its target account \sqlinline{(bankTgt, branchTgt, acctTgt)};
$R_5$ assigns the label \sqlinline{Transfer} to each edge in $R_2$;
finally, $R_6$ is empty in this example, i.e.\ $R_6=\varnothing$.

Using composite identifiers significantly simplifies pattern matching queries. For instance, one can succinctly write:

\[
   \left( (x) \overset{t}{\rightarrow}^{1,\infty}(y) \langle \theta \rangle \right)_{x.\texttt{bank},x.\texttt{branch},y.\texttt{bank},y.\texttt{branch}} \quad  \text{where} \quad 
    \theta \df \texttt{Transfer}(t)
\]

{The output of this query includes the identifiers of banks and branches of source and target accounts, which we can subsequently filter on without additional joins. 
By contrast, in Example~\ref{ex:pattern-transfer-interval} the output would have contained only the long {IBAN} strings for the endpoints, preventing such post filtering.}
\end{example}

\newcommand{\thmerwvsrw}{
  \(\erwpgq\) is strictly more expressive than \(\rwpgq\).
}

\begin{theorem}
\label{thm:erw_vs_rw}
\thmerwvsrw
\end{theorem}
The \erwpgq~fragment closes a gap in the expressiveness identified in~\cite{vldb25}. We discuss it in the following example:

\begin{example}[Increasing Values on Edges]\label{ex:ex-incr-edges}
Consider again the graph view defined in Example~\ref{ex:bank-account-ids}. 
Suppose we want to find all pairs of accounts connected by paths of transfers such that the transferred amounts strictly increase along the path edges. This was shown to be inexpressible by output patterns in~\cite{vldb25}.
We construct a new property graph in the $\erwpgq$ fragment as follows: For each  account node
\sqlinline{(bank,branch,acct)}, we create multiple  copies, for each incoming {amount} \sqlinline{l} transferred to it.
So now the identifier is \sqlinline{(bank,branch,acct,l)}. 

Edges in the new graph connect node copies in a strictly increasing manner: an edge 
connects node 
\sqlinline{(bank,branch,acct,l)}
to node 
\sqlinline{(bank,branch,acct,j)} 
only if \sqlinline{j}$ > $\sqlinline{l}. 
See Figure~\ref{fig:increasing-edges} for illustration.

Now, querying such paths is straightforward using standard path pattern matching. 
 We do not need to add a filter enforcing
 increasing amounts, the composite identifiers
already encode this information at every node copy, so any path produced by the
query is guaranteed to follow strictly increasing transfer amounts.

This demonstrates the expressive capability of the $\erwpgq$ fragment:
by enriching the natural  identifier 
we can materialize multiple copies of the same account, each copy
representing the node in a different state, determined by the incoming transfer amount. 
\end{example}

\begin{figure}[ht]
  \centering
  \begin{subfigure}[b]{0.49\linewidth}
    \centering
    \includegraphics[width=\linewidth]{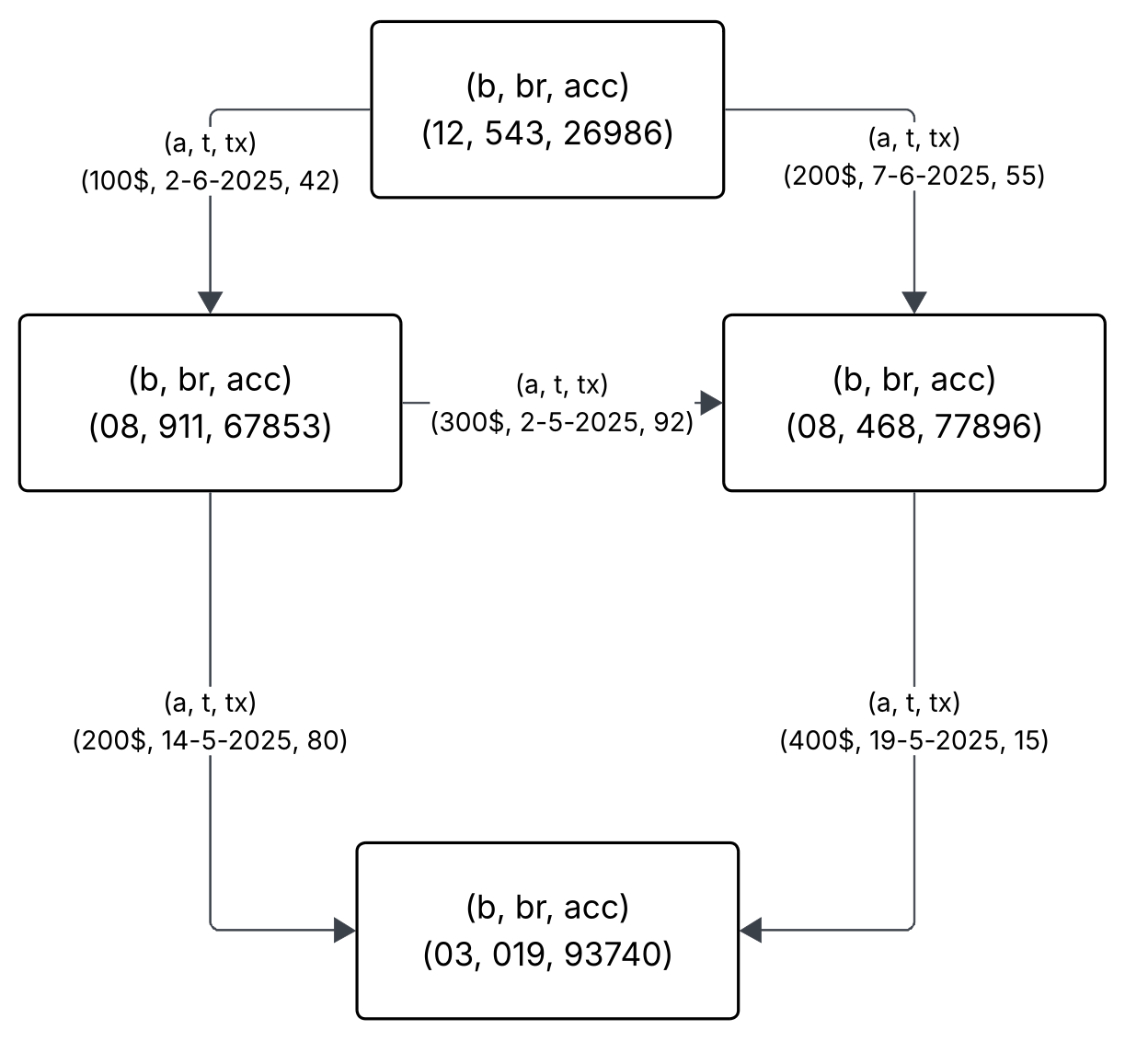}
    \caption{Input property graph $G$}
    \label{fig:inc-edges-input}
  \end{subfigure}
  \hfill
  \begin{subfigure}[b]{0.49\linewidth}
    \centering
    \includegraphics[width=\linewidth]{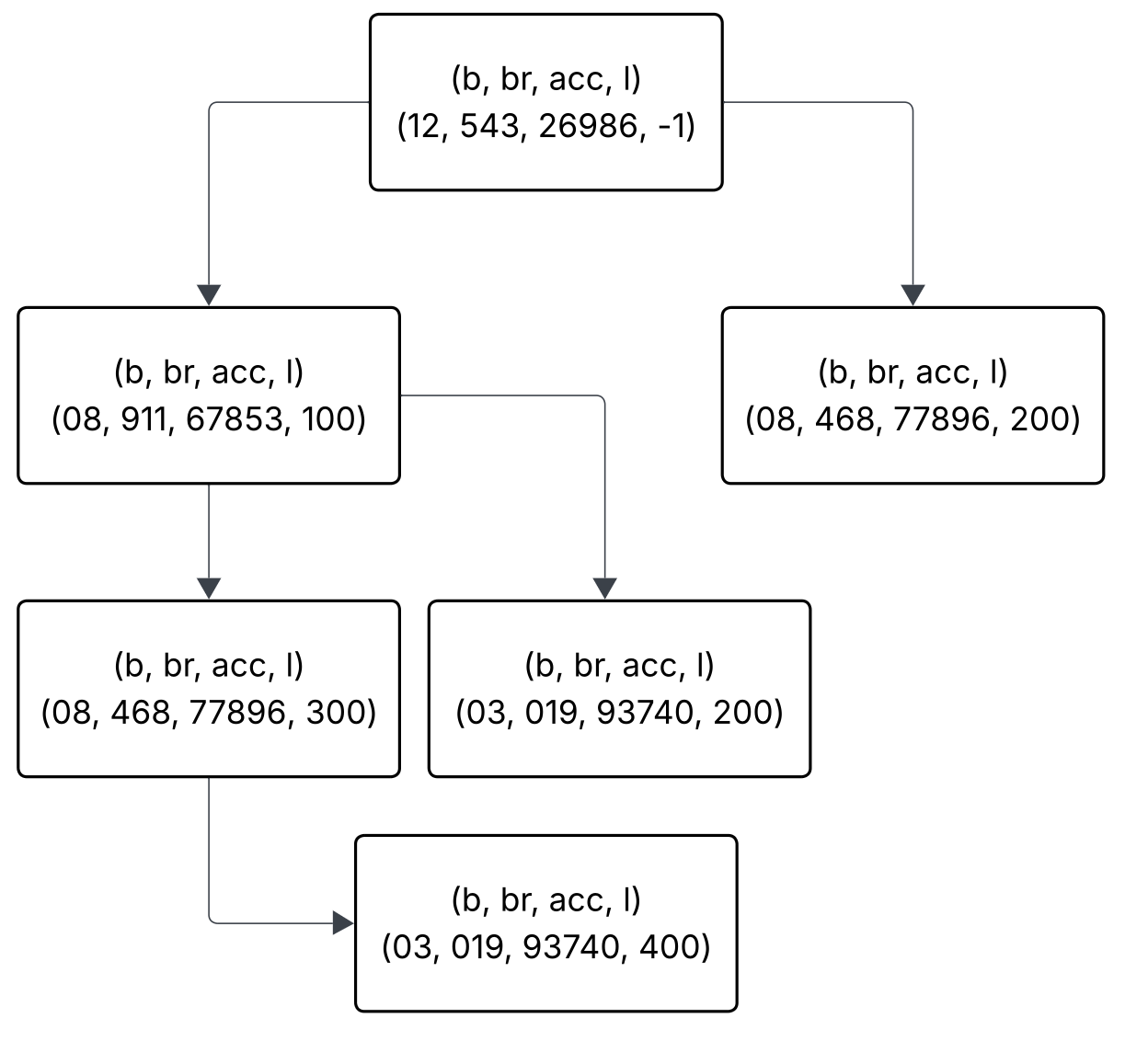}
    \caption{Constructed property graph $G'$}
    \label{fig:inc-edges-output}
  \end{subfigure}

  \caption{Illustration of the view-construction step in
    Example~\ref{ex:ex-incr-edges}.}
    \Description{Illustration of the view construction}
  \label{fig:increasing-edges}
\end{figure}

We now turn to a more fine-grained analysis of the language $\erwpgq$.

\section{Expressiveness Results}
\label{sec:expressiveness}

To characterize the expressive power of $\erwpgq$, we relate it to $\fotcc$. This connection establishes $\erwpgq$ as a robust graph query language and positions it precisely within the classical complexity class $\nlog$.

\subsection{\texorpdfstring{Equivalence of $\erwpgq$ and $\fotcc$}{Equivalence of ERWPGQ and FOTCC}}
\label{sec:capture}

Our goal now is to establish a precise correspondence between the expressive power of the extended read-write fragment~$\erwpgq$ and first-order logic with transitive closure, $\fotcc$.  
We proceed by formally recalling the semantics of $\fotcc$, and then prove the equivalence.

\paragraph{FO[TC] overview}

First-order (FO) formulas over a relational schema~$\schema$ are constructed from atomic formulas of the form $R(x_1,\ldots,x_n)$ where $R \in \schema$ has arity $n$, and from equalities $x = y$. Composite formulas are then built using the standard Boolean connectives $\lnot$, $\land$, $\lor$, along with the quantifiers $\exists$ and $\forall$.

Given a database instance~$\db$ over~$\schema$, and a tuple of values~$\bar a$ corresponding to a tuple of variables~$\bar x$, we write $\db \models \varphi[\bar a/\bar x]$ to indicate that the formula $\varphi(\bar x)$ is satisfied in~$\db$ when each free variable in~$\bar x$ is replaced by the corresponding value in~$\bar a$.

A tuple $\bar a$ satisfies
$\varphi(\bar x)$ in a database instance $\db$
(written $\db\models\varphi[\bar a/\bar x]$) when replacing the free
variables $\bar x$ by the concrete values $\bar a$ makes the sentence
true.

FO[TC] extends FO with the transitive-closure
operator
\[
  \mathrm{TC}_{(\bar u,\bar v)}\!\bigl[\psi(\bar u,\bar v,\bar p)\bigr]
  (\bar x,\bar y),
  \qquad|\bar u|=|\bar v|=|\bar x|=|\bar y|.
\]
The formula holds in a database~$\db$ under assignment \([\bar a, \bar b, \bar c / \bar u, \bar v, \bar p]\)
if there exists a finite sequence of tuples
\(
  \bar a = \bar a_0, \bar a_1, \ldots, \bar a_n = \bar b
\)
such that for each \(i < n\),
\(
  \db \models \psi[\bar a_i, \bar a_{i+1}, \bar c/\bar u, \bar v, \bar p].
\)
That is, \(\psi\) relates each consecutive pair in the sequence, using the same parameter tuple~\(\bar c\),
and the overall formula expresses reachability under the binary relation defined by~\(\psi\).

For any FO[TC] formula $\phi(\bar x)$ we write
\[
  \semd{\phi(\bar x)}
  \;:=\;
  \{\;\bar a \mid \db\models\phi[\bar a/\bar x]\;\},
\]
that is, the \emph{result relation} consisting of all tuples that make
$\phi$ true in~$\db$.

\newcommand{\thmerwpgqinfotc}{
For every schema $~\schema$ and for every query $Q \in \erwpgq$ over $\schema$ with $\arity{Q}=n$, there exists an $\mathrm{FO[TC]}$ formula
\( 
  \varphi_{Q}(x_1,\ldots, x_n)
\)
over the same schema $\schema$ such that
\[
  \semd{Q} = \semd{\varphi_{Q}(x_1,\ldots, x_n)}
\]
for every relational database $\db$ over~$\schema$.
}

With these definitions, we can now show the equivalence $\erwpgq = \mathrm{FO[TC]}$ by proving that each language subsumes the other. 

\begin{theorem}[$\erwpgq \subseteq \mathrm{FO[TC]}$]\label{thm:erwpgq-in-fotc}
\thmerwpgqinfotc
\end{theorem}

In the proof, we decompose the translation from $\erwpgq$ to $\fotcc$ into two steps.  
(i) Algebraic core: union, difference, Cartesian product, projection, and selection are mapped naturally to first-order connectives, laying the base of the induction.  
(ii) Pattern matching: we take a path pattern~$\pat$ and produce an \(\mathrm{FO[TC]}\) formula
\(
  \varphi_{\pat}(\bar{y},\bar{x}_{\src},\bar{x}_{\tgt})
\)
that incorporates two variables $\bar{x}_{\src},\bar{x}_{\tgt}$ for the source and target, respectively, of the matched paths, and
whose other free variables $\bar y$ explicitly track every binding in the semantics of~$\pat$. This translation coincides with the one in~\cite{FigueiraLinPeterfreund24RelPersp} up to minor notational differences.

The induction stitches these two parts (i) and (ii) together to prove the semantics equivalence.
As an illustration of step~(ii), consider the Kleene-star pattern
\(\pat^{*}\).
Assume we have the translation  $\tau(\pat)\bigl(\bar x,\bar u,\bar v\bigr)$
  of $\pat$ with free variables are $\bar x,\bar u,\bar v$. 
  Due to the semantics of path patterns (Figure~\ref{fig:sem-endpoints}), the translation $\tau\!\bigl(\pat^{*}\bigr)$ has free variables $\bar x_{\src},\bar x_{\tgt}$ and is defined recursively as follows:
\[
\tau\!\bigl(\pat^{*}\bigr)\!\bigl(\bar x_{\src}, \bar x_{\tgt}\bigr)
  \;=\;
     \exists \bar x:\, \Bigl(\mathrm{TC}_{\bar u,\bar v}\;
       \tau(\pat)\bigl(\bar x,\bar u,\bar v\bigr)
  \Bigr)\!\bigl(\bar x_{\src},\bar x_{\tgt}\bigr),
\]
This concrete instance demonstrates how a recursive graph pattern is
captured in \(\mathrm{FO[TC]}\) by wrapping the translation of the
base pattern \(\pat\) inside an appropriate transitive closure.

The other direction also holds as the next theorem shows.

\newcommand{\thmfotcinrerwpgq}{
For every schema $\,\schema$, and every $\mathrm{FO[TC]}$ formula
\(
  \varphi(x_1,\ldots, x_n)
\)
over $\,\schema$,
there exists a query
\(
  Q_{\varphi}\in\erwpgq
\)
over $\,\schema$ with $\arity{Q_{\varphi}}=n$ such that
\[
  \semd{\varphi(x_1,\ldots, x_n)}
  \;=\;
  \semd{Q_{\varphi}}
\]
for every relational database $\db$ over~$\,\schema$.
}

\newcommand{\lemtctranslation}{
For every schema $\schema$ and $\mathrm{FO[TC]}$~formula~\(\varphi(\bar u,\bar v,\bar p)\)~over~$\schema$
whose free-variable tuples are $\bar u=(u_{1},\dots,u_{k})$,
$\bar v=(v_{1},\dots,v_{k})$
and (possibly empty) parameters
$\bar p=(p_{1},\dots,p_{\ell})$.
Define the $\mathrm{FO[TC]}$ formula
\[
  \varphi^{\mathrm{TC}}(\bar x,\bar y,\bar p)
  \;:=\;
  \mathrm{TC}_{(\bar u,\bar v)}
     \bigl[\varphi(\bar u,\bar v,\bar p)\bigr]
     (\bar x,\bar y),
  \qquad
  |\bar x|=|\bar y|=|\bar u|=|\bar v|=k.
\]
 
There exists an \erwpgq~query~over~$\schema$
\[
  Q_{\varphi^{\mathrm{TC}}},
  \qquad
  \sch{Q_{\varphi^{\mathrm{TC}}}}
      \;=\;
      \bar x\cup\bar y\cup\bar p,
\]
such that for every relational database $\db$ over~$\schema$,
\[
  \semd{Q_{\varphi^{\mathrm{TC}}}}
  \;=\;
  \semd{\varphi^{\mathrm{TC}}(\bar x,\bar y,\bar p)}.
\]
}

\begin{theorem}[$\mathrm{FO[TC]} \subseteq \erwpgq$]
\label{thm:fotc-contained-in-erwpgq}
\thmfotcinrerwpgq
\end{theorem}

The proof is based on a full translation \(T\) from \(\mathrm{FO[TC]}\) to \(\erwpgq\).  
First-order connectives and quantifiers are translated inductively in a straightforward manner. The key step is the translation of transitive closure subformulas:
\[
  \psi
  \;=\;
  \mathrm{TC}_{(\bar u,\bar v)}
     \bigl[\varphi(\bar u,\bar v,\bar p)\bigr]
     (\bar x,\bar y),
  \qquad |\bar x|=|\bar y|=|\bar u|=|\bar v|.
\]

Note that the parameters \(\bar p\) remain fixed across all iterations of the transitive closure.  
By the inductive hypothesis, we obtain the translation \(T(\varphi)\) of \(\varphi\), and define
\(
  \mathit{C} := \pi_{\bar p}(T(\varphi)),
\)
the projection onto the parameter values.

Then, for each concrete parameter tuple \(\bar c \in \mathit{C}\),  
we isolate the corresponding \(\bar u, \bar v\) pairs and use the extended graph constructor from Section~\ref{sec:extensions}, namely $\pgViewExtendedUnion$, to define the new graph \(G_{\bar c}\).

We then apply the reachability pattern:
\(
  \pat_{\mathit{reach}}
  := \left((\bar x)\xrightarrow{}^{\!\!*}(\bar y)\right)_{\bar x, \bar y}
\)
to \(G_{\bar c}\). This returns all \((\bar x,\bar y)\) pairs connected by a path in \(G_{\bar c}\). We define the overall query:
\[
  Q_{\psi}
  :=\;
  \bigcup_{\bar c\in\mathit{C}}
    \Bigl(
      \pat_{\mathit{reach}}\!\bigl(G_{\bar c}\bigr)
      \times
      \{\bar c\}
    \Bigr),
  \qquad \sch{Q_{\psi}}=(\bar x,\bar y).
\]
 Note that this union is realized by an ordinary join
\(
  \pat_{\mathit{reach}}(G_{\bar c})
  \Join
  \sigma_{\bar p=\bar c}(\mathit{C})
\).
By construction, \(\semd{Q_{\psi}} = \semd{\psi}\) for every database \(\db\), so \(Q_{\psi}\) can be treated as an atomic subquery by outer Boolean connectives and quantifiers. 

Combining 
Theorems~\ref{thm:erwpgq-in-fotc} and ~\ref{thm:fotc-contained-in-erwpgq}, we obtain
\begin{corollary}[Expressive equivalence]
\label{cor:erwpgq-equals-fotc}
In terms of expressiveness, 
$
  \erwpgq
  \;=\;
  \mathrm{FO[TC]}.
$
\end{corollary}
This leads to the following straightforward observation. 
    \begin{corollary}[\nlog~capture]\label{cor:erwpgq-equals-nl}
The evaluation problem for $\erwpgq$ is \nlog~complete.
\end{corollary}

We now dive deeper into the language 
$\erwpgq{}$ to better understand what it is that enables it to capture every query definable in 
$\fotcc$.

\subsection{\texorpdfstring{Inside $\erwpgq$}{Inside ERWPGQ}}
\label{subsec:inside-erwpgq}
What makes $\erwpgq$ more expressive than $\rwpgq$?  
To answer this, we adopt a more fine-grained analysis of the language's capabilities.  
Recall from Definition~\ref{def:pgView-n} the function
\[
  \pgViewExtendedUnion \;=\; \bigcup_{n\ge 1} \pgView^{=n},
\]
which allows node and edge identifiers of \emph{any} arity.  
Higher arities enable richer view definitions, which in turn unlock additional expressive power not available in the unary-identifier setting of $\rwpgq$.

For every fixed integer \(n\ge 1\), we define $\pgqextn$~as the language obtained from $\erwpgq$~by replacing 
$\patExtended{\return}(\bar\query)$ with 
$\patExtendedN{\return}(\bar\query)$, and set the semantics to be 
\[
  \semd{\patExtendedN{\return}(\bar\query)}
    := \sem{\pat_\return}_G,
    \quad\text{where }G = \pgView^{n}
        \bigl(
          \semd{\query_1},\ldots,
          \semd{\query_6}
        \bigr)
\]
and \(\bar\query \;:=\; (\query_1, \query_2, \query_3, \query_4, \query_5, \query_6).\)

On the $\fotcc$ side, we write $\fotcn$ to denote the fragment of first-order logic with transitive closure in which all transitive closure operators 
\[
  \mathrm{TC}_{\bar{u},\bar{v}}[\psi(\bar{u},\bar{v},\bar{p})](\bar{x},\bar{y})
\]
use tuples $\bar{u}, \bar{v}, \bar{x}, \bar{y}$ of fixed arity $n$.

With this notation, we can show that $\pgqextn$ translates to $\fotcn$.

\newcommand{\thmpgqextninfotcn}{
Let \(n\ge 1\).  
For every schema~\(\schema\) and every query  
\(Q \in \pgqextn\) over~\(\schema\) with \(\arity{Q} = k\),  
there exists an \(\fotcn\) formula  
\(
  \varphi_{Q}(x_{1},\dots,x_{k})
\)  
such that  
\[
  \semd{Q}  
  \;=\;  
  \semd{\varphi_{Q}(x_{1},\dots,x_{k})}  
\]  
for every relational database~\(\db\) over~\(\schema\).
}

\begin{theorem}[\(\pgqextn \subseteq \fotcn\)]
\label{thm:pgqextn-in-fotcn}
\thmpgqextninfotcn
\end{theorem}

The converse translation also exists. 
\newcommand{\thmfotcninpgqextn}{
Let \(n \ge 1\).  
For every schema~\(\schema\) and every  
\(\fotcn\) formula \(\varphi(x_1,\ldots, x_k)\) over~\(\schema\),  
there exists a query  
\(Q_{\varphi} \in \pgqextn\) over~\(\schema\)  
with \(\arity{Q_{\varphi}} = k\) such that  
\[
  \semd{\varphi(x_1,\ldots, x_k)}  
  \;=\;  
  \semd{Q_{\varphi}}  
\]  
for every relational database~\(\db\) over~\(\schema\).
}
\begin{theorem}[\(\fotcn \subseteq \pgqextn\)]
\label{thm:fotcn-in-pgqextn}
\thmfotcninpgqextn
\end{theorem}

Both Theorems~\ref{thm:pgqextn-in-fotcn} and~\ref{thm:fotcn-in-pgqextn} are obtained by inductive syntax-directed translations.

{These are refined versions of the unary translations already given in
Theorems~\ref{thm:erwpgq-in-fotc} and~\ref{thm:fotc-contained-in-erwpgq} to the \(n\)-ary case.
For \(\pgqextn \subseteq \fotcn\), we  generalize the translation
from
Theorem~\ref{thm:erwpgq-in-fotc} 
     to ensure the resulting formula is in 
     \(\tc{n}\).
For \(\fotcn \subseteq \pgqextn\), we {generalize} the translation
from
        Theorem~\ref{thm:fotc-contained-in-erwpgq} to 
        ensure the resulting query is in \pgqextn.
This gives desired equivalence \(\pgqextn = \fotcn\).}

Combining
Theorems~\ref{thm:pgqextn-in-fotcn}
and~\ref{thm:fotcn-in-pgqextn},
we obtain the following.
\begin{corollary}[Expressive equivalence for fixed arity]
\label{cor:pgqextn-equals-fotcn}
In terms of expressiveness, 
\[
  \pgqextn
  \;=\;
  \fotcn
  \qquad
  \text{for every } n\ge 1.
\]
\end{corollary}

We can derive the following expressiveness chain on fragments within $\erwpgq$.

\begin{theorem}[Expressiveness chain]\label{thm:one-two-chain}
\[
  \rwpgq
  \;=\;
  \pgqext{1}
  \;=\;
  \fotc{1}
  \;\subsetneq\;
  \fotc{2}
  \;=\;
  \fotcn
  \;=\;
  \erwpgq.
\]
\end{theorem}
\begin{proof}
The chain begins with the observation that the read-write fragment is, by definition, the same as the arity-1 extension; hence $\rwpgq=\pgqext{1}$.  Corollary~\ref{cor:pgqextn-equals-fotcn} then equates this fragment with unary transitive-closure logic, giving $\pgqext{1}=\fotc{1}$. By Immerman’s collapse theorem for ordered structures, unary TC is strictly weaker than binary TC, and every higher arity collapses to the binary case, formally, $\fotc{1}\subsetneq\fotc{2}=\fotcn$ for all $n\ge 2$~\cite{Immerman1999}.  
The same corollary yields $\fotc{2}=\pgqext{2}$, so every construct definable in any higher-arity fragment $\pgqext{n}$ can already be written in $\pgqext{2}$. Because $\erwpgq$ allows identifiers of arity n, we have $\pgqext{n}=\erwpgq$.  Combining these equalities gives the desired chain.
\end{proof}

\section{Comparison with the SQL/PGQ Standard}
\label{subsec:std-gap}
We now clarify how our theoretical framework relates to the SQL/PGQ standard.  

\paragraph{Property Graph Views}
The ISO draft~\cite{sqlpgq-standard} distinguishes a \emph{pure property graph (PPG)} from its
underlying \emph{tabular property graph (TPG)}.  
In our framework, PPG coincides with the property graph of
Definition~\ref{def:pg}, whereas TPG with property graph view of Definition~\ref{def:pgView}.
We represent property graph views by the
canonical relations \((R_1,\ldots,R_6)\) that form the inputs to the
\textsf{pgView} family (Definitions~\ref{def:pgView}-\ref{def:pgView-n-func}).  
This choice differs from the Standard’s
\emph{vertex-table/edge-table} notation:
\begin{enumerate}
  \item \textit{Normalization:} The six-relation encoding we use is {highly normalized}: each elementary fact, that is, 
vertex or edge identifier, source, target, label, or property, occurs only once.
By comparison, a single vertex table can repeat the same label and property values
across many rows, and an edge table can replicate source or target keys. {Strictly adopting the vertex-edge-table layout of the ISO draft may affect our results.}
  \item  \textit{Identifiers:}
        The standard requires each vertex and edge row to carry a primary
        key. 
        This is reflected by $R_1$ (vertex identifiers) and
        $R_2$ (edge identifiers) in our model, ensuring that query
        results can always return the graph elements, that is, its node and edge IDs.
  \item  \textit{Arity of Identifiers:}
       Our core read-write fragment $\rwpgq$ assumes
        {unary identifiers}, matching the simplest TPGs whose keys
        are single columns.  Our extended fragment $\erwpgq$
        generalizes this to \(n\)-ary identifiers, covering $n$-ary
        keys allowed by the standard but keeping them outside the core
        language. This separation reflects the usual treatment of relational algebra: formal models of core SQL deal only with single, atomic attributes, whereas the full SQL standard goes further and also allows composite (multi-attribute) types.
        \item 
        \textit{Labels:}
        Unlike the standard, our model does not require every edge to have a label. Under our definitions, an edge may be unlabeled, that is, it can appear in the graph without any associated label entry in $R_5$. To align fully with the ISO draft, one may simply insert a distinguished default label  for such edges.
\end{enumerate}

\paragraph{Updates.}
We omit update operations on TPGs to keep a clean theoretical core. This causes no loss of generality: any change can be simulated by rebuilding the six base relations and reapplying~$\pgView$ (or its extensions), so all results in this paper still hold.

\section{Conclusion and Future Work}
\label{sec:conclusion}
We have characterized the expressive power of SQL/PGQ, showing how its graph view mechanism fundamentally shapes the language’s capabilities. Our analysis identifies three fragments, $\rpgq$, $\rwpgq$, and $\erwpgq$, and places them in a strict expressiveness hierarchy culminating in $\nlog$, captured precisely by the extended fragment $\erwpgq$. While $\erwpgq$ matches this foundational benchmark, it is not yet clear how naturally its expressiveness translates into usable query syntax. For example, the ``increasing values on edges'' query
provably inexpressible in the pattern matching layer~\cite{vldb25}, can be expressed in $\erwpgq$ using composite identifiers and dynamic graph construction as described in Example~\ref{ex:ex-incr-edges}.  However, the required encoding is nontrivial.
To express that query, we simulate value tracking by creating node copies, each annotated with the amount of the incoming edge. Edges connect only copies where the target amount is greater than the amount of the source. This is expressible in $\erwpgq$ and demonstrates its ability to capture complex value-based recursion but also highlights that the encoding is intricate and unlikely to be intuitive for users.
Whether such queries can be formulated more naturally within $\erwpgq$, or whether new syntax is needed to express them cleanly, remains an open question. Finally, our formalization opens the door to compositional graph-query languages: $\pgView$ constructs full property graphs that can be queried or outputted. In addition to compositionality, our model offers a natural foundation for unified languages that query both relations and graph databases, and can move naturally and seamlessly from one model to the other.

\begin{acks}
The authors were supported by ISF grant 2355/24.
\end{acks}

\newpage
\bibliographystyle{ACM-Reference-Format}
\bibliography{Arxiv/references}


\begin{thebibliography}{25}


\ifx \showCODEN    \undefined \def \showCODEN     #1{\unskip}     \fi
\ifx \showDOI      \undefined \def \showDOI       #1{#1}\fi
\ifx \showISBNx    \undefined \def \showISBNx     #1{\unskip}     \fi
\ifx \showISBNxiii \undefined \def \showISBNxiii  #1{\unskip}     \fi
\ifx \showISSN     \undefined \def \showISSN      #1{\unskip}     \fi
\ifx \showLCCN     \undefined \def \showLCCN      #1{\unskip}     \fi
\ifx \shownote     \undefined \def \shownote      #1{#1}          \fi
\ifx \showarticletitle \undefined \def \showarticletitle #1{#1}   \fi
\ifx \showURL      \undefined \def \showURL       {\relax}        \fi
\providecommand\bibfield[2]{#2}
\providecommand\bibinfo[2]{#2}
\providecommand\natexlab[1]{#1}
\providecommand\showeprint[2][]{arXiv:#2}

\bibitem[Angles et~al\mbox{.}(2018)]%
        {gcore}
\bibfield{author}{\bibinfo{person}{Renzo Angles}, \bibinfo{person}{Marcelo Arenas}, \bibinfo{person}{Pablo Barcel{\'o}}, \bibinfo{person}{Peter~A. Boncz}, \bibinfo{person}{George H.~L. Fletcher}, \bibinfo{person}{Claudio Gutierrez}, \bibinfo{person}{Tobias Lindaaker}, \bibinfo{person}{Marcus Paradies}, \bibinfo{person}{Stefan Plantikow}, \bibinfo{person}{Juan~F. Sequeda}, \bibinfo{person}{Oskar van Rest}, {and} \bibinfo{person}{Hannes Voigt}.} \bibinfo{year}{2018}\natexlab{}.
\newblock \showarticletitle{G-CORE: A Core for Future Graph Query Languages}. In \bibinfo{booktitle}{\emph{Proceedings of the 2018 International Conference on Management of Data (SIGMOD '18)}}. \bibinfo{publisher}{Association for Computing Machinery}, \bibinfo{address}{New York, NY, USA}, \bibinfo{pages}{1421--1432}.
\newblock
\urldef\tempurl%
\url{https://doi.org/10.1145/3183713.3190654}
\showDOI{\tempurl}


\bibitem[Angles et~al\mbox{.}(2017)]%
        {surveyChile}
\bibfield{author}{\bibinfo{person}{Renzo Angles}, \bibinfo{person}{Marcelo Arenas}, \bibinfo{person}{Pablo Barcel{\'{o}}}, \bibinfo{person}{Aidan Hogan}, \bibinfo{person}{Juan~L. Reutter}, {and} \bibinfo{person}{Domagoj Vrgoc}.} \bibinfo{year}{2017}\natexlab{}.
\newblock \showarticletitle{Foundations of Modern Query Languages for Graph Databases}.
\newblock \bibinfo{journal}{\emph{{ACM} Comput. Surv.}} \bibinfo{volume}{50}, \bibinfo{number}{5} (\bibinfo{year}{2017}), \bibinfo{pages}{68:1--68:40}.
\newblock


\bibitem[Barcel{\'{o}} et~al\mbox{.}(2012)]%
        {BarceloLLW-tods12}
\bibfield{author}{\bibinfo{person}{Pablo Barcel{\'{o}}}, \bibinfo{person}{Leonid Libkin}, \bibinfo{person}{Anthony~Widjaja Lin}, {and} \bibinfo{person}{Peter~T. Wood}.} \bibinfo{year}{2012}\natexlab{}.
\newblock \showarticletitle{Expressive languages for path queries over graph-structured data}.
\newblock \bibinfo{journal}{\emph{{ACM} Trans. Database Syst.}} \bibinfo{volume}{37}, \bibinfo{number}{4} (\bibinfo{year}{2012}), \bibinfo{pages}{31:1--31:46}.
\newblock


\bibitem[Barcel{\'{o}} et~al\mbox{.}(2014)]%
        {BarceloLR-jacm14}
\bibfield{author}{\bibinfo{person}{Pablo Barcel{\'{o}}}, \bibinfo{person}{Leonid Libkin}, {and} \bibinfo{person}{Juan~L. Reutter}.} \bibinfo{year}{2014}\natexlab{}.
\newblock \showarticletitle{Querying regular graph patterns}.
\newblock \bibinfo{journal}{\emph{Journal of the {ACM}}} \bibinfo{volume}{61}, \bibinfo{number}{1} (\bibinfo{year}{2014}), \bibinfo{pages}{8:1--8:54}.
\newblock


\bibitem[Barcel{\'o} et~al\mbox{.}(2012)]%
        {NRE}
\bibfield{author}{\bibinfo{person}{Pablo Barcel{\'o}}, \bibinfo{person}{Jorge P{\'e}rez}, {and} \bibinfo{person}{Juan~L. Reutter}.} \bibinfo{year}{2012}\natexlab{}.
\newblock \showarticletitle{Relative Expressiveness of Nested Regular Expressions}. In \bibinfo{booktitle}{\emph{Proceedings of the 6th Alberto Mendelzon International Workshop on Foundations of Data Management (AMW 2012)}} \emph{(\bibinfo{series}{CEUR Workshop Proceedings}, Vol.~\bibinfo{volume}{866})}, \bibfield{editor}{\bibinfo{person}{Juliana Freire} {and} \bibinfo{person}{Dan Suciu}} (Eds.). \bibinfo{publisher}{CEUR-WS.org}, \bibinfo{address}{Aachen, Germany}, \bibinfo{pages}{180--195}.
\newblock
\urldef\tempurl%
\url{https://ceur-ws.org/Vol-866/paper13.pdf}
\showURL{%
\tempurl}


\bibitem[Calvanese et~al\mbox{.}(2000)]%
        {C2RPQ}
\bibfield{author}{\bibinfo{person}{Diego Calvanese}, \bibinfo{person}{Giuseppe~De Giacomo}, \bibinfo{person}{Maurizio Lenzerini}, {and} \bibinfo{person}{Moshe~Y. Vardi}.} \bibinfo{year}{2000}\natexlab{}.
\newblock \showarticletitle{Containment of Conjunctive Regular Path Queries with Inverse}. In \bibinfo{booktitle}{\emph{Proceedings of the Seventh International Conference on Principles of Knowledge Representation and Reasoning (KR 2000)}}, \bibfield{editor}{\bibinfo{person}{Anthony~G. Cohn}, \bibinfo{person}{Fausto Giunchiglia}, {and} \bibinfo{person}{Bart Selman}} (Eds.). \bibinfo{publisher}{Morgan Kaufmann}, \bibinfo{address}{San Francisco, CA, USA}, \bibinfo{pages}{176--185}.
\newblock
\newblock
\shownote{Breckenridge, CO, USA, April 11--15, 2000}.


\bibitem[Cruz et~al\mbox{.}(1987)]%
        {RPQ}
\bibfield{author}{\bibinfo{person}{Isabel~F. Cruz}, \bibinfo{person}{Alberto~O. Mendelzon}, {and} \bibinfo{person}{Peter~T. Wood}.} \bibinfo{year}{1987}\natexlab{}.
\newblock \showarticletitle{A Graphical Query Language Supporting Recursion}. In \bibinfo{booktitle}{\emph{Proceedings of the Association for Computing Machinery Special Interest Group on Management of Data 1987 Annual Conference, San Francisco, CA, USA, May 27-29, 1987}}, \bibfield{editor}{\bibinfo{person}{Umeshwar Dayal} {and} \bibinfo{person}{Irving~L. Traiger}} (Eds.). \bibinfo{publisher}{{ACM} Press}, \bibinfo{address}{New York, NY, USA}, \bibinfo{pages}{323--330}.
\newblock
\urldef\tempurl%
\url{https://doi.org/10.1145/38713.38749}
\showDOI{\tempurl}


\bibitem[Figueira et~al\mbox{.}(2024)]%
        {FigueiraLinPeterfreund24RelPersp}
\bibfield{author}{\bibinfo{person}{Diego Figueira}, \bibinfo{person}{Anthony~W. Lin}, {and} \bibinfo{person}{Liat Peterfreund}.} \bibinfo{year}{2024}\natexlab{}.
\newblock \showarticletitle{Relational Perspective on Graph Query Languages}.
\newblock \bibinfo{journal}{\emph{CoRR}}  \bibinfo{volume}{abs/2407.06766} (\bibinfo{year}{2024}).
\newblock
\urldef\tempurl%
\url{https://doi.org/10.48550/arXiv.2407.06766}
\showDOI{\tempurl}
\showeprint[arXiv]{2407.06766}


\bibitem[Francis et~al\mbox{.}(2023a)]%
        {pods23}
\bibfield{author}{\bibinfo{person}{Nadime Francis}, \bibinfo{person}{Am{\'{e}}lie Gheerbrant}, \bibinfo{person}{Paolo Guagliardo}, \bibinfo{person}{Leonid Libkin}, \bibinfo{person}{Victor Marsault}, \bibinfo{person}{Wim Martens}, \bibinfo{person}{Filip Murlak}, \bibinfo{person}{Liat Peterfreund}, \bibinfo{person}{Alexandra Rogova}, {and} \bibinfo{person}{Domagoj Vrgoc}.} \bibinfo{year}{2023}\natexlab{a}.
\newblock \showarticletitle{{GPC:} {A} Pattern Calculus for Property Graphs}. In \bibinfo{booktitle}{\emph{Proceedings of the 42nd {ACM} {SIGMOD-SIGACT-SIGAI} Symposium on Principles of Database Systems, {PODS} 2023, Seattle, WA, USA, June 18-23, 2023}}, \bibfield{editor}{\bibinfo{person}{Floris Geerts}, \bibinfo{person}{Hung~Q. Ngo}, {and} \bibinfo{person}{Stavros Sintos}} (Eds.). \bibinfo{publisher}{{ACM}}, \bibinfo{address}{Seattle, WA, USA}, \bibinfo{pages}{241--250}.
\newblock
\urldef\tempurl%
\url{https://doi.org/10.1145/3584372.3588662}
\showDOI{\tempurl}


\bibitem[Francis et~al\mbox{.}(2023b)]%
        {icdt23}
\bibfield{author}{\bibinfo{person}{Nadime Francis}, \bibinfo{person}{Am{\'{e}}lie Gheerbrant}, \bibinfo{person}{Paolo Guagliardo}, \bibinfo{person}{Leonid Libkin}, \bibinfo{person}{Victor Marsault}, \bibinfo{person}{Wim Martens}, \bibinfo{person}{Filip Murlak}, \bibinfo{person}{Liat Peterfreund}, \bibinfo{person}{Alexandra Rogova}, {and} \bibinfo{person}{Domagoj Vrgoc}.} \bibinfo{year}{2023}\natexlab{b}.
\newblock \showarticletitle{A Researcher's Digest of {GQL}}. In \bibinfo{booktitle}{\emph{26th International Conference on Database Theory, {ICDT} 2023, March 28-31, 2023, Ioannina, Greece}} \emph{(\bibinfo{series}{LIPIcs}, Vol.~\bibinfo{volume}{255})}, \bibfield{editor}{\bibinfo{person}{Floris Geerts} {and} \bibinfo{person}{Brecht Vandevoort}} (Eds.). \bibinfo{publisher}{Schloss Dagstuhl - Leibniz-Zentrum f{\"{u}}r Informatik}, \bibinfo{address}{Dagstuhl, Germany}, \bibinfo{pages}{1:1--1:22}.
\newblock
\urldef\tempurl%
\url{https://doi.org/10.4230/LIPICS.ICDT.2023.1}
\showDOI{\tempurl}


\bibitem[Gaifman(1982)]%
        {GAIFMAN1982105}
\bibfield{author}{\bibinfo{person}{Haim Gaifman}.} \bibinfo{year}{1982}\natexlab{}.
\newblock \showarticletitle{On Local and Non-Local Properties}.
\newblock In \bibinfo{booktitle}{\emph{Proceedings of the Herbrand Symposium}}, \bibfield{editor}{\bibinfo{person}{J.~Stern}} (Ed.). \bibinfo{series}{Studies in Logic and the Foundations of Mathematics}, Vol.~\bibinfo{volume}{107}. \bibinfo{publisher}{Elsevier}, \bibinfo{address}{Amsterdam}, \bibinfo{pages}{105--135}.
\newblock
\showISSN{0049-237X}
\urldef\tempurl%
\url{https://doi.org/10.1016/S0049-237X(08)71879-2}
\showDOI{\tempurl}


\bibitem[Gaifman and Vardi(1985)]%
        {GaifmanVardi85}
\bibfield{author}{\bibinfo{person}{Haim Gaifman} {and} \bibinfo{person}{Moshe~Y. Vardi}.} \bibinfo{year}{1985}\natexlab{}.
\newblock \showarticletitle{A simple proof that connectivity of finite graphs is not first-order definable}.
\newblock \bibinfo{journal}{\emph{Bull. {EATCS}}}  \bibinfo{volume}{26} (\bibinfo{year}{1985}), \bibinfo{pages}{43--44}.
\newblock


\bibitem[Gheerbrant et~al\mbox{.}(2025)]%
        {vldb25}
\bibfield{author}{\bibinfo{person}{Am{\'e}lie Gheerbrant}, \bibinfo{person}{Leonid Libkin}, \bibinfo{person}{Liat Peterfreund}, {and} \bibinfo{person}{Alexandra Rogova}.} \bibinfo{year}{2025}\natexlab{}.
\newblock \showarticletitle{GQL and SQL/PGQ: Theoretical Models and Expressive Power}.
\newblock \bibinfo{journal}{\emph{Proceedings of the VLDB Endowment (PVLDB)}} \bibinfo{volume}{18}, \bibinfo{number}{6} (\bibinfo{year}{2025}), \bibinfo{pages}{1798--1810}.
\newblock
\showISSN{2150-8097}
\urldef\tempurl%
\url{https://doi.org/10.14778/3725688.3725707}
\showDOI{\tempurl}
\newblock
\shownote{Licensed under CC BY-NC-ND 4.0}.


\bibitem[{GQL Standards Committee}(2024)]%
        {GQLStandards}
\bibfield{author}{\bibinfo{person}{{GQL Standards Committee}}.} \bibinfo{year}{2024}\natexlab{}.
\newblock \bibinfo{title}{GQL Standards Website}.
\newblock
\newblock
\urldef\tempurl%
\url{https://www.gqlstandards.org/}
\showURL{%
\tempurl}
\newblock
\shownote{Accessed: November 2024}.


\bibitem[Green et~al\mbox{.}(2021)]%
        {LDBC:TR:TR-2021-01}
\bibfield{author}{\bibinfo{person}{Alastair Green}, \bibinfo{person}{Paolo Guagliardo}, {and} \bibinfo{person}{Leonid Libkin}.} \bibinfo{year}{2021}\natexlab{}.
\newblock \bibinfo{booktitle}{\emph{Property graphs and paths in GQL: Mathematical definitions}}.
\newblock \bibinfo{type}{Technical Reports} TR-2021-01. \bibinfo{institution}{Linked Data Benchmark Council (LDBC)}.
\newblock
\urldef\tempurl%
\url{https://doi.org/10.54285/ldbc.TZJP7279}
\showDOI{\tempurl}


\bibitem[Grädel and McColm(1996)]%
        {GRADEL1996169}
\bibfield{author}{\bibinfo{person}{Erich Grädel} {and} \bibinfo{person}{Gregory~L. McColm}.} \bibinfo{year}{1996}\natexlab{}.
\newblock \showarticletitle{Hierarchies in transitive closure logic, stratified Datalog and infinitary logic}.
\newblock \bibinfo{journal}{\emph{Annals of Pure and Applied Logic}} \bibinfo{volume}{77}, \bibinfo{number}{2} (\bibinfo{year}{1996}), \bibinfo{pages}{169--199}.
\newblock
\showISSN{0168-0072}
\urldef\tempurl%
\url{https://doi.org/10.1016/0168-0072(95)00021-6}
\showDOI{\tempurl}


\bibitem[Immerman(1999)]%
        {Immerman1999}
\bibfield{author}{\bibinfo{person}{Neil Immerman}.} \bibinfo{year}{1999}\natexlab{}.
\newblock \bibinfo{booktitle}{\emph{Descriptive Complexity}}.
\newblock \bibinfo{publisher}{Springer}, \bibinfo{address}{New York}.
\newblock
\showISBNx{978-0-387-94703-1}


\bibitem[{ISO/IEC JTC 1/SC 32/WG 3}(2023)]%
        {sqlpgq-standard}
\bibfield{author}{\bibinfo{person}{{ISO/IEC JTC 1/SC 32/WG 3}}.} \bibinfo{year}{2023}\natexlab{}.
\newblock \bibinfo{title}{{ISO/IEC 9075-16: SQL/PGQ - Property Graph Queries}}.
\newblock \bibinfo{howpublished}{\url{https://www.iso.org/standard/79473.html}}.
\newblock
\newblock
\shownote{Working Draft, ISO/IEC JTC 1/SC 32/WG 3}.


\bibitem[Libkin et~al\mbox{.}(2025)]%
        {LibkinMMPV25}
\bibfield{author}{\bibinfo{person}{Leonid Libkin}, \bibinfo{person}{Wim Martens}, \bibinfo{person}{Filip Murlak}, \bibinfo{person}{Liat Peterfreund}, {and} \bibinfo{person}{Domagoj Vrgoc}.} \bibinfo{year}{2025}\natexlab{}.
\newblock \showarticletitle{Querying Graph Data: Where We Are and Where To Go}. In \bibinfo{booktitle}{\emph{Companion of the 44th Symposium on Principles of Database Systems, {PODS} 2025, Berlin, Germany, June 22-27, 2025}}, \bibfield{editor}{\bibinfo{person}{Floris Geerts} {and} \bibinfo{person}{Benny Kimelfeld}} (Eds.). \bibinfo{publisher}{{ACM}}, \bibinfo{address}{New York, NY, USA}, \bibinfo{pages}{9--26}.
\newblock
\urldef\tempurl%
\url{https://doi.org/10.1145/3722234.3725822}
\showDOI{\tempurl}


\bibitem[Mirza et~al\mbox{.}(2003)]%
        {mirza2003studying}
\bibfield{author}{\bibinfo{person}{Batul~J Mirza}, \bibinfo{person}{Benjamin~J Keller}, {and} \bibinfo{person}{Naren Ramakrishnan}.} \bibinfo{year}{2003}\natexlab{}.
\newblock \showarticletitle{Studying recommendation algorithms by graph analysis}.
\newblock \bibinfo{journal}{\emph{Journal of intelligent information systems}}  \bibinfo{volume}{20} (\bibinfo{year}{2003}), \bibinfo{pages}{131--160}.
\newblock


\bibitem[openCypher(2017)]%
        {openCypher}
\bibfield{author}{\bibinfo{person}{openCypher}.} \bibinfo{year}{2017}\natexlab{}.
\newblock \bibinfo{title}{Cypher Query Language Reference, Version 9}.
\newblock
\newblock
\urldef\tempurl%
\url{https://github.com/opencypher/openCypher/blob/master/docs/openCypher9.pdf}
\showURL{%
\tempurl}


\bibitem[Reutter et~al\mbox{.}(2015)]%
        {reutter2017rq}
\bibfield{author}{\bibinfo{person}{Juan~L. Reutter}, \bibinfo{person}{Miguel Romero}, {and} \bibinfo{person}{Moshe~Y. Vardi}.} \bibinfo{year}{2015}\natexlab{}.
\newblock \showarticletitle{Regular Queries on Graph Databases}. In \bibinfo{booktitle}{\emph{18th International Conference on Database Theory (ICDT 2015)}} \emph{(\bibinfo{series}{Leibniz International Proceedings in Informatics (LIPIcs)}, Vol.~\bibinfo{volume}{31})}, \bibfield{editor}{\bibinfo{person}{Marcelo Arenas} {and} \bibinfo{person}{Martín Ugarte}} (Eds.). \bibinfo{publisher}{Schloss Dagstuhl -- Leibniz-Zentrum für Informatik}, \bibinfo{address}{Brussels, Belgium}, \bibinfo{pages}{177--194}.
\newblock
\urldef\tempurl%
\url{https://doi.org/10.4230/LIPIcs.ICDT.2015.177}
\showDOI{\tempurl}


\bibitem[Srivastava and Singh(2023)]%
        {srivastava2023fraud}
\bibfield{author}{\bibinfo{person}{Sakshi Srivastava} {and} \bibinfo{person}{Anil~Kumar Singh}.} \bibinfo{year}{2023}\natexlab{}.
\newblock \showarticletitle{Fraud detection in the distributed graph database}.
\newblock \bibinfo{journal}{\emph{Cluster Computing}} \bibinfo{volume}{26}, \bibinfo{number}{1} (\bibinfo{year}{2023}), \bibinfo{pages}{515--537}.
\newblock


\bibitem[Stonebraker and Pavlo(2024)]%
        {WhatGoesAround}
\bibfield{author}{\bibinfo{person}{Michael Stonebraker} {and} \bibinfo{person}{Andrew Pavlo}.} \bibinfo{year}{2024}\natexlab{}.
\newblock \showarticletitle{What Goes Around Comes Around... And Around...}
\newblock \bibinfo{journal}{\emph{SIGMOD Rec.}} \bibinfo{volume}{53}, \bibinfo{number}{2} (\bibinfo{date}{July} \bibinfo{year}{2024}), \bibinfo{pages}{21--37}.
\newblock
\showISSN{0163-5808}
\urldef\tempurl%
\url{https://doi.org/10.1145/3685980.3685984}
\showDOI{\tempurl}


\bibitem[van Rest et~al\mbox{.}(2016)]%
        {PGQL}
\bibfield{author}{\bibinfo{person}{Oskar van Rest}, \bibinfo{person}{Sungpack Hong}, \bibinfo{person}{Jinha Kim}, \bibinfo{person}{Xuming Meng}, {and} \bibinfo{person}{Hassan Chafi}.} \bibinfo{year}{2016}\natexlab{}.
\newblock \showarticletitle{PGQL: a Property Graph Query Language}. In \bibinfo{booktitle}{\emph{Proceedings of the Fourth International Workshop on Graph Data Management Experiences and Systems (GRADES '16)}}. \bibinfo{publisher}{Association for Computing Machinery}, \bibinfo{address}{New York, NY, USA}, Article \bibinfo{articleno}{7}, \bibinfo{numpages}{6}~pages.
\newblock
\urldef\tempurl%
\url{https://doi.org/10.1145/2960414.2960421}
\showDOI{\tempurl}


\end{thebibliography}

\newpage
\section{Appendix}
\subsection{Appendix for Section~\ref{sec:prelims}}

\begin{figure*}[b]
\setlength{\tabcolsep}{0pt}
\fbox{\makebox[\dimexpr\linewidth-8pt\relax][l]{%
$
\begin{array}{@{}rl@{}}
\multicolumn{2}{@{}l}{\textbf{Patterns (path semantics):}}\\[2pt]
\sem{(x)}^{\mathrm{path}}_G
  &\!:=\!
    \bigl\{\, (p,\{x\mapsto n\}) \mid n\in N,\; \src(p)=\tgt(p)=n \bigr\}
\\
\sem{\overset{x}{\rightarrow}}^{\mathrm{path}}_G
  &\!:=\!
    \bigl\{\, (p,\{x\mapsto e\})
          \mid e\in E,\; \src(e)=\src(p),\; \tgt(e)=\tgt(p)\bigr\}
\\
\sem{\overset{x}{\leftarrow}}^{\mathrm{path}}_G
  &\!:=\!
    \bigl\{\, (p,\{x\mapsto e\})
          \mid e\in E,\; \src(e)=\src(p),\; \tgt(e)=\tgt(p)\bigr\}
\\
\sem{\pat_1 + \pat_2}^{\mathrm{path}}_G
  &\!:=\!
    \sem{\pat_1}^{\mathrm{path}}_G \cup \sem{\pat_2}^{\mathrm{path}}_G
\\
\sem{\pat_1\,\pat_2}^{\mathrm{path}}_G
  &\!:=\!
    \bigl\{\, (p_1\!\cdot\! p_2,\;\mu_1 \join \mu_2)
          \,\big|\,
          (p_1,\mu_1)\in\sem{\pat_1}^{\mathrm{path}}_G,\;
          (p_2,\mu_2)\in\sem{\pat_2}^{\mathrm{path}}_G,
\\[-2pt]
  &\hphantom{\!:=\!\bigl\{\, (p_1\!\cdot\! p_2,\;\mu_1 \join \mu_2)\,\big|\,}
          \tgt(p_1)=\src(p_2),\;
          \mu_1 \sim \mu_2
    \bigr\}
\\
\sem{\pat\langle\theta\rangle}^{\mathrm{path}}_G
  &\!:=\!
    \bigl\{\, (p,\mu)\in\sem{\pat}^{\mathrm{path}}_G \mid \mu\models\theta \bigr\}
\\
\sem{\pat^{n..m}}^{\mathrm{path}}_G
  &\!:=\!
    \displaystyle\bigcup_{i=n}^{m} \sem{\pat}^{\,i,\mathrm{path}}_G \quad {\text{where}\quad
  \sem{\pat}^{\,0,\mathrm{path}}_G := \bigl\{\, (p,\mu_\emptyset) \mid \src(p)=\tgt(p) \bigr\}, \text{and for } n>0}
\\
\sem{\pat}^{\,n,\mathrm{path}}_G   &\!:=\! \bigl\{
  (p,\mu_\emptyset)
  \ \Bigm|\
  \exists \mu_{1},\ldots,\mu_{n}, p_{1},\ldots,p_{n}:\;
  (p_{i},\mu_{i}) \in \sem{\pat}^{\mathrm{path}}_G,\ 
\\[-2pt]
  &\hphantom{\!:=\! \bigl\{(p,\mu_\emptyset)\ \Bigm|\ } 
  \tgt(p_{i})=\src(p_{i+1})\ \forall i<n,\ 
  p=p_{1}\cdot\ldots\cdot p_{n}
  \bigr\}
\\[2pt]
\multicolumn{2}{@{}l}{\textbf{Output Patterns:}}\\[2pt]
\sem{\pat_\return}^{\mathrm{path}}_G
  &\!:=\!
    \bigl\{\,
(\mu_\return(\omega_1),\ldots , \mu_\return(\omega_n))
    \mid \exists p:\,(p,\mu)\in\sem{\pat}^{\mathrm{path}}_G,\ \return \df (\omega_1,\ldots,\omega_n) \bigr\}
\end{array}
$
}}
\caption{Semantics of Patterns and Output Patterns (path semantics).} 
\Description{Semantics of graph patterns using path semantics.}
\label{fig:sem-path}
\end{figure*}
Let \(G\) be a property graph and \(\pat\) a pattern. Following~\cite{vldb25}, we denote the semantics of \(\pat\) on \(G\) by \(\sem{\pat}^{\mathrm{path}}_{G}\), as illustrated in Figure~\ref{fig:sem-path}.

The following equivalence has been used in prior work on path query semantics~\cite{FigueiraLinPeterfreund24RelPersp}. Here we state it and prove it formally.
\begin{proposition}[Endpoint-path equivalence]
\label{prop:endpoint-equiv}
Let $G$ be a property graph and $\pat$ a pattern.
Define the projection
\[
\pi_{\mathrm{end}}\!\bigl((p,\mu)\bigr)\;=\;
\bigl(\src(p),\tgt(p),\mu\bigr).
\]
Then
\[
\pi_{\mathrm{end}}\!\bigl(\sem{\pat}^{\mathrm{path}}_{G}\bigr)
\;=\;
\sem{\pat}_{G}.
\]
In particular, $\sem{\pat}^{\mathrm{path}}_{G}$ and
$\sem{\pat}_{G}$ return identical result relations when they
are plugged into the relational-algebra layer of core PGQ.
\end{proposition}
\begin{proof}
We prove the equality by showing both inclusions.
\paragraph{\normalfont($\subseteq$)\;}
Let $(p,\mu)\in\sem{\pat}^{\mathrm{path}}_{G}$.  We show by
structural induction on~$\pat$ that
$\bigl(\src(p),\tgt(p),\mu\bigr)\in\sem{\pat}_{G}$.
Base Cases:
\begin{description}
  \item[\textsc{(T1) Node}] $\pat=(x)$.
    The path semantics yields the single vertex path
    $v$ with $\mu(x)=v$.  The endpoint rule produces
    $(v,v,\mu)$, which is exactly $\pi_{\mathrm{end}}((v,\mu))$.
  \item[\textsc{(T2) Forward Edge}] $\pat=(x\!\rightarrow y)$.
    The path semantics yields the one edge path $p=e$
    with $\src(e)=v$, $\tgt(e)=w$
    and $\mu(x)=v,\;\mu(y)=w$.
    The endpoint rule admits $(v,w,\mu)$, again matching
    $\pi_{\mathrm{end}}((e,\mu))$.
  \item[\textsc{(T3) Backward Edge}] $\pat=(x\!\leftarrow y)$. Symmetric.
\end{description}
Induction Steps:
\begin{description}
 \item[\textsc{(T4) Concatenation}] $\pat=\pat_{1}\,\pat_{2}$.
    From path semantics we have
    $(p_{1},\mu_{1})\in\sem{\pat_{1}}^{\mathrm{path}}_{G}$,
    $(p_{2},\mu_{2})\in\sem{\pat_{2}}^{\mathrm{path}}_{G}$,
    $\tgt(p_{1})=\src(p_{2})$, and $\mu_{1}\sim\mu_{2}$,
    with $p=p_{1}p_{2}$ and $\mu=\mu_{1}\cup\mu_{2}$.
    By the induction hypothesis,
    \[
      (s,m,\mu_{1})\in\sem{\pat_{1}}_{G},
      \qquad
      (m,t,\mu_{2})\in\sem{\pat_{2}}_{G}
    \]
    where $s=\src(p_{1})$ and $t=\tgt(p_{2})$.
    Since $m$ coincides and $\mu_{1}\sim\mu_{2}$, the
    definition of the endpoint semantics for concatenation
    adds $(s,t,\mu)$, which equals $\pi_{\mathrm{end}}((p,\mu))$.
    \item[\textsc{(T5) Disjunction}] $\pat=\pat_{1}+\pat_{2}$.  
    Path semantics reaches $(p,\mu)$ by choosing either
    $\pat_{1}$ or~$\pat_{2}$:
    \begin{enumerate}
      \item If $(p,\mu)\in\sem{\pat_{1}}^{\mathrm{path}}_{G}$,
            the induction hypothesis gives
            $\pi_{\mathrm{end}}((p,\mu))\in
            \sem{\pat_{1}}_{G}$.
      \item If $(p,\mu)\in\sem{\pat_{2}}^{\mathrm{path}}_{G}$,
            the induction hypothesis gives the same inclusion with
            $\pat_{2}$.
    \end{enumerate}
    In either sub-case, the definition of the endpoint semantics for
    $\pat_{1}+\pat_{2}$ simply takes the union of the two
    component semantics, so the triple
    $\pi_{\mathrm{end}}((p,\mu))$ is included.
      \item[\textsc{(T6) Bounded repetition}]$\pat=\pat_{0}^{n..m}$.  
        In path semantics there exist integers
        $k$ with $n\le k\le m$ and derivations
        \[
           (p_{1},\mu_{1}),\dots,(p_{k},\mu_{k})
           \in\sem{\pat_{0}}^{\mathrm{path}}_{G}
        \]
        \[
           \quad\text{such that}\quad
           p=p_{1}\cdots p_{k},\;
           \mu=\bigcup_{i=1}^{k}\mu_{i},
           \; \mu_{1}\sim\cdots\sim\mu_{k},
           \;\tgt(p_{i})=\src(p_{i+1}).
        \]
        By the induction hypothesis each
        $\pi_{\mathrm{end}}((p_{i},\mu_{i}))=
        (s_{i},t_{i},\mu_{i})$ belongs to
        $\sem{\pat_{0}}_{G}$.
        The equalities above give
        $t_{i}=s_{i+1}$ for $1\le i<k$,
        so the definition of the endpoint semantics for
        $\pat_{0}^{n..m}$ contributes the triple
        \(
          (s_{1},t_{k},\mu)
        \),
        which is precisely
        $\pi_{\mathrm{end}}((p,\mu))$.

  \item[\textsc{(T7) Filtering}]$\pat=\pat_{0}\langle\theta\rangle$.  
        From path semantics we have
        $(p,\mu)\in\sem{\pat_{0}}^{\mathrm{path}}_{G}$
        and $G,\mu\models\theta$.
        The inductive step places
        $(s,t,\mu)=\pi_{\mathrm{end}}((p,\mu))$
        in $\sem{\pat_{0}}_{G}$,
        and because the same satisfaction test
        $G,\mu\models\theta$ appears unchanged in the endpoint
        definition, $(s,t,\mu)$ is also in
        $\sem{\langle\theta\rangle\pat_{0}}_{G}$.

    \end{description}

    \paragraph{\normalfont($\supseteq$)\;}
    Now assume $(s,t,\mu)\in\sem{\pat}_{G}$.
    We construct $(p,\mu)\in\sem{\pat}^{\mathrm{path}}_{G}$ with
    $\src(p)=s$ and $\tgt(p)=t$ by
    structural induction on~$\pat$.
    
    Base Cases:
    \begin{description}
      \item[\textsc{(T1) Node}] $\pat=(x)$.\;
        Endpoint semantics enforces $s=t=\mu(x)$.
        Let $p$ be the single vertex path $s$.
        Then $(p,\mu)\in\sem{(x)}^{\mathrm{path}}_{G}$ and
        $\pi_{\mathrm{end}}((p,\mu))=(s,t,\mu)$.

      \item[\textsc{(T2) Forward Edge}] $\pat=(x\!\rightarrow y)$.\;
        By definition of the endpoint rule there exists an edge
        $e$ with $\src(e)=s$, $\tgt(e)=t$,
        $\mu(x)=s$ and $\mu(y)=t$.
        Choosing $p=e$ yields the desired outcome.

      \item[\textsc{(T3) Backward Edge}] $\pat=(x\!\leftarrow y)$.\;
        Symmetric: take the edge $e$ with
        $\src(e)=t$ and $\tgt(e)=s$, set $p=e$.
    \end{description}
    Induction Steps:
    \begin{description}
      \item[\textsc{(T4) Concatenation}] $\pat=\pat_{1}\,\pat_{2}$.\;
        Endpoint semantics provides
        \[
  (s,m,\mu_{1})\in\sem{\pat_{1}}_{G},
  \qquad
  (m,t,\mu_{2})\in\sem{\pat_{2}}_{G},
  \qquad
  \mu_{1}\sim\mu_{2},
  \qquad
  \mu=\mu_{1}\cup\mu_{2}.
\]
        By the induction hypothesis we obtain paths
        $(p_{1},\mu_{1})$ and $(p_{2},\mu_{2})$
        whose concatenation $p=p_{1}p_{2}$ lies in
        $\sem{\pat}^{\mathrm{path}}_{G}$ and has the
        required endpoints.

      \item[\textsc{(T5) Disjunction}] $\pat=\pat_{1}+\pat_{2}$.\;
        Because endpoint semantics is a union,
        $(s,t,\mu)$ originates from either $\pat_{1}$ or $\pat_{2}$.
        Apply the induction hypothesis to that branch to
        obtain the corresponding path witness.

      \item[\textsc{(T6) Bounded repetition}] $\pat=\pat_{0}^{n..m}$.\;
        We have integers $k$ with $n\le k\le m$ and triples
        \[
          (s_{1},s_{2},\mu_{1}),\dots,(s_{k},s_{k+1},\mu_{k})
          \in\sem{\pat_{0}}_{G}
        \]
        such that $s_{1}=s$, $s_{k+1}=t$,
        $\mu=\bigcup_{i=1}^{k}\mu_{i}$ and
        $\mu_{1}\sim\cdots\sim\mu_{k}$.
        The induction hypothesis gives paths
        $(p_{i},\mu_{i})\in\sem{\pat_{0}}^{\mathrm{path}}_{G}$
        with endpoints $s_{i}\to s_{i+1}$.
        Concatenating $p=p_{1}\cdots p_{k}$ yields
        $(p,\mu)\in\sem{\pat}^{\mathrm{path}}_{G}$ and
        $\pi_{\mathrm{end}}((p,\mu))=(s,t,\mu)$.

      \item[\textsc{(T7) Filtering}] $\pat=\pat_{0}\langle\theta\rangle$.\;
        Endpoint semantics gives
        $(s,t,\mu)\in\sem{\pat_{0}}_{G}$
        and $G,\mu\models\theta$.
        The induction hypothesis supplies a path
        $(p,\mu)\in\sem{\pat_{0}}^{\mathrm{path}}_{G}$,
        we have $G,\mu\models\theta$ in the
        path rule, so $(p,\mu)\in\sem{\pat}^{\mathrm{path}}_{G}$
        and $\pi_{\mathrm{end}}((p,\mu))=(s,t,\mu)$.
    \end{description}

    Hence every triple produced by endpoint semantics has
    a matching path witness, establishing
    \[
      \sem{\pat}_{G}
      \;\subseteq\;
      \pi_{\mathrm{end}}\!\bigl(\sem{\pat}^{\mathrm{path}}_{G}\bigr).
    \]

    Combining the inclusion $\sem{\pat}_{G}
          \subseteq
          \pi_{\mathrm{end}}\!\bigl(\sem{\pat}^{\mathrm{path}}_{G}\bigr)$
    established above with the reverse inclusion
    $\pi_{\mathrm{end}}\!\bigl(\sem{\pat}^{\mathrm{path}}_{G}\bigr)
          \subseteq
          \sem{\pat}_{G}$ proven at the beginning of the proof,
    we conclude that
    \[
      \pi_{\mathrm{end}}\!\bigl(\sem{\pat}^{\mathrm{path}}_{G}\bigr)
      \;=\;
      \sem{\pat}_{G}.
    \]
\end{proof}

\subsection{Appendix for Section~\ref{sec:rwpgq}}

Consider a relational database $\db_G$ representing a graph whose node identifiers are partitioned into two disjoint relations, one for red nodes and one for blue nodes. That is, the database contains two disjoint relations \texttt{RedNodes} and  \texttt{BlueNodes}, each specifying a set of red and blue node IDs, respectively. 
In addition, the database contains a single relation \texttt{Edges} storing edge identifiers, and two relations \texttt{Source} and \texttt{Target} that associate each edge with its source and target node IDs, respectively, such that
\begin{itemize}
    \item $\texttt{Source}(e) \in \texttt{BlueNodes}$ implies $\texttt{Target}(e) \in \texttt{RedNodes}$, and
    \item 
    $\texttt{Target}(e) \in \texttt{BlueNodes}$ implies $\texttt{Source}(e) \in \texttt{RedNodes}$.
\end{itemize}
This means that edges are from red nodes to blue ones, or from blue to red. In particular there are no edges whose both ends are of the same color.

We use \(\mathsf{RA}\) to denote relational algebra queries, that is, the fragment of  \(\rpgq\)  obtained by omitting the pattern matching construct.
\begin{proposition}
   For every query $\query \in \rpgq$ there is an equivalent $\query' \in \mathsf{RA}$ 
   on $\db_G$. 
   That is, 
    \[
    \sem{\query}_{\db_G} = \sem{\query'}_{\db_G}.
    \]
\end{proposition}
\begin{proof}
This is proved  by induction on $\query$. The only non-trivial case is $\query \df \pat_{\return}(\bar R)$.
    (All other cases are in $\mathsf{RA}$.)
    Notice that by definition $R_1$ is either \texttt{BlueNodes} or \texttt{RedNodes} (and these are symmetric cases), and $R_2$ can only be \texttt{Edges}.
However, due to the definition of $\pgView$ there are not any valid options for $R_3$ and $R_4$. Hence the claim holds trivially.    
\end{proof}

\begin{reptheorem}{\ref{thm:rwsupro}}
    \thmrwsupro
\end{reptheorem}

\begin{proof}
    The inclusion \( \rpgq~ \subseteq \rwpgq  \) follows directly. To establish that this containment is {strict}, we present a concrete example.
Suppose we want to detect \emph{alternating-color paths}, i.e.\ paths of any length whose nodes alternate red-blue at every step.  
Assume, for contradiction, that some query \(Q\in\rpgq\) returns \textsf{true} exactly when such a path exists.  
Because \rpgq~is non-recursive, \(Q\) is an FO expression over the six canonical relations produced by \textsf{pgView}.  
By Gaifman-Vardi locality~\cite{GaifmanVardi85}, every FO formula can inspect the input only within a finite radius; equivalently, \(Q\) can check paths of length at most \(f(|Q|)\) for some function \(f\).  
Take a database that contains an alternating path longer than this bound.  
The property holds, but \(Q\) cannot see it, contradicting the assumption.  
Hence alternating-color paths are \emph{not} expressible in \rpgq.

In contrast, $\rwpgq$~allows the dynamic construction of graph views. In particular, the first argument of $\bar \query$ in a query of the form $\pat_{\return}(\bar \query)$ may define the union of the red and blue node relations, enabling the construction of a graph that includes all nodes and edges. { Below we show how to construct the union graph view and issue the query:  
  Let
  \[
     \bar Q
       \;=\;
       \bigl(
        \texttt{RedNodes} \cup
   \texttt{BlueNodes},\;
        \texttt{Edges},\;
         \texttt{Source},\;
        \texttt{Target},\;
         \emptyset,\;
         \emptyset
       \bigr).
  \]
  With this tuple of six relational subqueries, the alternating color path
  query is expressed as
        $\pat_{\emptyset}(\bar Q)$
        with $\pat$ being the reachability pattern $()\rightarrow^*()$
        }
\end{proof}

\begin{reptheorem}{\ref{thm:rw_subset_nl}}
   \thmrwsubsetnl
\end{reptheorem}
\begin{proof}
We demonstrate this limitation using the same example from \cite{vldb25}. Specifically, consider the Boolean query that determines whether there exists a path whose length does \emph{not} form a semilinear set. Such a query is known to be in \nlog, yet we show that it is not expressible in \rwpgq.

The key observation is that the set of path lengths detectable by \rwpgq~can be characterized in Presburger Arithmetic (the first-order theory of $(\mathbb{N}, +, <)$).
Since Presburger-definable sets are precisely the semilinear sets, \rwpgq~is inherently limited to expressing queries over semilinear path lengths. Therefore, any query requiring detection of non-semilinear lengths lies beyond the expressive power of \rwpgq.
\end{proof}
\subsection{Appendix for Section~\ref{sec:extensions}}
\begin{reptheorem}{\ref{thm:erw_vs_rw}}
\thmerwvsrw
\end{reptheorem}

\begin{proof}
By definition, \(\rwpgq\subseteq\erwpgq\).
To show that this is a strict containment, 
we show that there exists a query definable in \erwpgq\ but not in
\rwpgq.  The argument relies on the following notation and  results whose complete proofs appear in
Section~\ref{sec:expressiveness}:

\begin{enumerate}
  \item Corollary~\ref{cor:pgqextn-equals-fotcn} proves that the unary fragment
        \(\pgqext{1}\) (which coincides with \rwpgq) captures exactly the logic
        \(\fotc{1}\) of first-order formulas with unary-arity transitive-closure
        operators.
  \item {According to~\cite{Immerman1999, GRADEL1996169}, we have the following inclusion: 
\(
  \fotc{1} \subsetneq \fotc{2} = \fotcn \) for every \(n\ge 2
\).\footnote{Notice that on unordered structures this hierarchy does not collapse in $\tc{2}$, that is, \(\fotc{1} \subsetneq \fotc{2} \subsetneq \fotc{3} \subsetneq \cdots.\)).}}
    \item  Corollary~\ref{cor:erwpgq-equals-fotc} establishes
        \(\erwpgq = \mathrm{FO[TC]}\), so every \(\fotc{2}\)-formula is
        expressible by some \(\pgqext{2}\subseteq\erwpgq\) query.
\end{enumerate}

We combine these pieces together as follows.  
Choose any formula
\[
  \phi^{(2)}(x_1,x_2,y_1,y_2)\;:=\;
  TC_{(u_1,u_2),(v_1,v_2)}\bigl[E(u_1,u_2,v_1,v_2)\bigr]
  \bigl((x_1,x_2),(y_1,y_2)\bigr)
\]
that defines reachability over pairs of nodes.  
Item~(ii) says \(\phi^{(2)}\notin\fotc{1}\).
By item~(iii) there exists a \(\pgqext{2}\)-query \(Q_{\mathsf{pairs}}\in\erwpgq\)
whose semantics coincides with \(\phi^{(2)}\).

Assume towards contradiction that \(Q_{\mathsf{pairs}}\) were also in
\rwpgq.  
Item (i) would then translate it into a formula of \(\fotc{1}\),
contradicting \(\phi^{(2)}\notin\fotc{1}\).
Hence \(Q_{\mathsf{pairs}}\notin\rwpgq\), and the inclusion
\(\rwpgq\subseteq\erwpgq\) is strict, and
thus, \(\erwpgq\) is strictly more expressive than \(\rwpgq\).
\end{proof}
\subsection{Appendix for Section~\ref{sec:expressiveness}}\label{apx:appendix_section_6}

\newcommand{
\lempatterntranslation}{
For every schema$~\schema$, for every path pattern $\pat$ with $\sch{\pat}=\bar{x}$, and for every $(Q_{1},\dots,Q_{6})$ with $Q_i \in \erwpgq$ ($1 \le i \le 6$), there exists an $\mathrm{FO[TC]}$ formula
\[
  \varphi_{\pat}(\bar{x},\bar{x}_{\src},\bar{x}_{\tgt}),
  \qquad
  \sch{\varphi_{\pat}}=
  (\bar{x}, \bar{x}_{\src},\bar{x}_{\tgt}),
\]
such that for every relational database $\db$ over $\schema$, and
\[
  G \;=\;
  \pgViewExtendedUnion\!\bigl(
      \semd{Q_{1}},\dots,\semd{Q_{6}}
  \bigr),
\]
the following hold:
\begin{enumerate}\renewcommand{\labelenumi}{(\arabic{enumi})}
  \item\textbf{Soundness:}\;
        $\displaystyle
        \forall(\bar s,\bar t,\mu)\in\sem{\pat}_{G}\;
        \exists\mu'\in\semd{\varphi_{\pat}}:\;
          \mu\sim_{\bar{x}}\mu',\;
          \bar s=\mu'(\bar{x}_{\src}),\;
          \bar t=\mu'(\bar{x}_{\tgt}).$
  \item\textbf{Completeness:}\;
        $\displaystyle
        \forall\mu'\in\semd{\varphi_{\pat}}\;
        \exists\mu\;:\;
          \bigl(\mu'(\bar{x}_{\src}),\mu'(\bar{x}_{\tgt}),\mu\bigr)
          \in\sem{\pat}_{G},\;
          \mu\sim_{\bar{x}}\mu'.$
\end{enumerate}
Throughout, an assignment
\(\mu\) is assumed to satisfy
\(\dom{\mu}=\bar{x}\), while an assignment
\(\mu'\) satisfies
\(\dom{\mu'}=(\bar{x}, \bar{x}_{\src},\bar{x}_{\tgt})\).
The tuples \(\bar{x}_{\src}\) and \(\bar{x}_{\tgt}\) have the same
arity as the identifier in
\(\semd{Q_{1}}=\semd{Q_{2}}\).
}

\begin{lemma}[Pattern matching translation]\label{lem:pattern-translation}
\lempatterntranslation
\end{lemma}

\begin{proof}
We construct by structural recursion on the pattern
$\pat$ a translation
\[
  \tau\colon\text{path patterns}\pfunc\mathrm{FO[TC]},
  \qquad
  \tau(\pat)=
    \phi_{\pat}\bigl(\barx,\bar{x}_{\src},\bar{x}_{\tgt}\bigr),
\]
and then verify the required \emph{soundness} and \emph{completeness}
conditions.

For every sub-pattern we display the defining
$\mathrm{FO[TC]}$-formula of~$\tau$.
Free variables are always
$(\bar{x}, \bar{x}_{\src},\bar{x}_{\tgt})$.
{
By definition, the six relations used throughout the proof are the
components of this graph view:
\[
\begin{aligned}
  N      \df \semd{Q_{1}},   E \df \semd{Q_{2}},
  \src \df \semd{Q_{3}}, 
  \tgt \df \semd{Q_{4}},
  \lab \df \semd{Q_{5}},
  \prop  \df \semd{Q_{6}}.
\end{aligned}
\]
Unless otherwise stated the tuple of free variables
\(\sch{\pat}\)
is kept implicit.}

\begin{tabular}{@{}ll@{}}
\!\!\textit{(T1) Node} $(\bar{v})$\!: &
  $\tau\bigl((\bar{v})\bigr)\;{}\df{}\;
      N(v)\;\land\;
      \bar{v}=\bar{x}_{\src}\;\land\;
      \bar{x}_{\src}=\bar{x}_{\tgt}$ \\[4pt]

\!\!\textit{(T2) Forward-edge} $\overset{e}{\rightarrow}$\!: &
  $\tau\bigl(\overset{e}{\rightarrow})\;\df\;
      E(e)\;\land\;
      \src(e,\bar{x}_{\src})\;\land\;
      \tgt(e,\bar{x}_{\tgt})$ \\[4pt]

\!\!\textit{(T3) Backward-edge} $\overset{e}{\leftarrow}$\!: &
  $\tau\bigl(\overset{e}{\leftarrow})\;\df\;
      E(e)\;\land\;
      \src(e,\bar{x}_{\tgt})\;\land\;
      \tgt(e,\bar{x}_{\src})$ \\[4pt]
\!\!\textit{(T4) Concatenation} $\pat_{1}\concat\pat_{2}$\!: &
   $\tau(\pat_{1}\concat\pat_{2})\;\df\;
    \exists \bar y\;
      \bigl(
        \tau(\pat_{1})(\barx_{1},\bar{x}_{\src},\bar y)\;\land\;
        \tau(\pat_{2})(\barx_{2}, \bar y,\bar{x}_{\tgt})
     \bigr)$ \\[4pt]

\!\!\textit{(T5) Disjunction} $\pat_{1}+\pat_{2}$\!: &
  $\tau(\pat_{1}+\pat_{2})\;\df\;
     \tau(\pat_{1})\;\lor\;\tau(\pat_{2})$ \\[4pt]


\!\!\textit{(T6) Bounded repetition} $\pat^{n..m}$\\[-2pt]
\hphantom{\textit{(T6) Bounded repetition}}\!$(0\!\le n\!\le m)$\!: &
$\displaystyle
\begin{aligned}[t]
  \tau(\pat^{n..m}) &\df \bigvee_{r=n}^{m}\tau(\pat^{r}),\\
  \tau(\pat^{0})    &\df (\bar{x}_{\src}=\bar{x}_{\tgt}),\;
  \tau(\pat^{r+1})  &\df \tau(\pat\concat\pat^{r})
\end{aligned}$ \\[8pt]

\!\!\textit{(T7) Filtering} $\pat\langle\theta\rangle$\!: &
  $\tau\!\bigl(\pat\langle\theta\rangle\bigr)\;\df\;
     \tau(\pat)\;\land\;
     \theta^{\mathrm{FO}}(\barx,\bar{x}_{\src},\bar{x}_{\tgt})$ \\[4pt]

\!\!\textit{(T8) Kleene star} $\pat^{*}$\!: &
  $\displaystyle
  \tau(\pat^{*})(\bar{x}_{\src},\bar{x}_{\tgt})\;\df\;
 \exists \bar x:  \bigl(
     \mathrm{TC}_{\bar{u},\bar{v}}\;
       \tau(\pat)(\barx,\bar{u},\bar{v})
   \bigr)\!(\bar{x}_{\src},\bar{x}_{\tgt})$
\end{tabular}

In (T7) the formula $\theta^{\mathrm{FO}}$ is the standard first-order
translation of the filter condition, using the base relations
$N,E,\lab,\src,\tgt,\prop$.
{
Let $\theta$ be a filter drawn from the grammar of
Figure \ref{fig:grammar-core-pgq}.  
We recursively build an $\mathrm{FO}[\mathrm{TC}]$ formula
\(
  \theta^{\mathrm{FO}}\bigl(\bar{x},\bar{x}_{\mathsf{src}},\bar{x}_{\mathsf{tgt}}\bigr)
\)
whose free variables are precisely the pattern variables $\bar{x}$ together
with the distinguished endpoints $\bar{x}_{\mathsf{src}},\bar{x}_{\mathsf{tgt}}$.

We have:
\[
\begin{array}{rcl}
  (x.k = x'.k')^{\mathrm{FO}}
     &:=& \displaystyle
     \exists v\,v'.\;
       \prop(x,k,v)\;\land\;\prop(x',k',v')\;\land\;v=v' \\[6pt]
  \bigl(\ell(x)\bigr)^{\mathrm{FO}}
     &:=& \lab(x,\ell) \\[6pt]
  (\neg\theta)^{\mathrm{FO}}
     &:=& \lnot\,\theta^{\mathrm{FO}} \\[6pt]
  (\theta_1 \wedge \theta_2)^{\mathrm{FO}}
     &:=& \theta_1^{\mathrm{FO}} \;\land\; \theta_2^{\mathrm{FO}} \\[6pt]
  (\theta_1 \vee \theta_2)^{\mathrm{FO}}
     &:=& \theta_1^{\mathrm{FO}} \;\lor\; \theta_2^{\mathrm{FO}}
\end{array}
\]
And these are all the options for $\theta^{\mathrm{FO}}$.}

{Note that in (T8), here, and in what follows, we treat $\pat^{*}$ as the unbounded repetition, i.e., $\pat^{*}\equiv\pat^{0..\infty}$. In our FO[TC] translation, (T6) covers only the bounded case $0\le n\le m<\infty$ via a finite disjunction, whereas the unbounded case requires transitive closure, hence (T8) handles $\pat^{*}$ separately.}

Let
\[
  \gdb\; \df\;
  \pgViewExtendedUnion\!\bigl(
      \semd{Q_{1}},\dots,\semd{Q_{6}}
  \bigr).
\]
We prove Soundness and Completeness by induction on the structure of~$\pat$:

\emph{Base cases} (T1)-(T3) follow by immediate inspection of the
semantics of nodes and edges.

\emph{Inductive cases.}
\begin{description}
\item[\textit{(T4) Concatenation} $\pat_{1}\concat\pat_{2}$.]
\mbox{}\par\addvspace{2pt}

\begin{enumerate}
\item \emph{Soundness.}\;
  {
  Suppose $(\bar s, \bar t,\mu)\in\sem{\pat_{1}\concat\pat_{2}}_{G}$.
  By the semantic clause for concatenation there exist
  $u,\mu_{1},\mu_{2}$ such that
  \(
      (\bar s, \bar u,\mu_{1})\in\sem{\pat_{1}}_{G},
      (\bar u, \bar t,\mu_{2})\in\sem{\pat_{2}}_{G},
      \mu=\mu_{1}\cup\mu_{2}.
  \)
  The inductive hypothesis applied to \(\pat_{1}\) and \(\pat_{2}\)
  yields assignments
  \(
      \mu'_{1}\in\semd{\phi_{\pat_{1}}},
      \mu'_{2}\in\semd{\phi_{\pat_{2}}},
  \)
  such that
  \(
      \mu_{1}\sim_{\bar{x}_{1}}\mu'_{1},\;
      \mu_{2}\sim_{\bar{x}_{2}}\mu'_{2},\;
      \mu'_{1}(\bar{x}_{\src})=s,\;
      \mu'_{1}(y)=u=\mu'_{2}(y),\;
      \mu'_{2}(\bar{x}_{\tgt})=t.
  \)
  Because \(\mu'_{1}\) and \(\mu'_{2}\) agree on the shared
  variable~\(y\), their union
  \(
    \mu' := \mu'_{1}\cup\mu'_{2}
  \)
  is well defined and satisfies
  \(
       \mu' \in \semd{\,\exists y\,
           \bigl(\phi_{\pat_{1}}(\bar{x}_{1},\bar{x}_{\src},y)\land
                 \phi_{\pat_{2}}(\bar{x}_{2},y,\bar{x}_{\tgt})\bigr)}
       \;=\;
       \semd{\phi_{\pat_{1}\concat\pat_{2}}}.
  \)
  Moreover \(\mu\sim_{\bar{x}}\mu'\), and
  \(s=\mu'(\bar{x}_{\src})\), \(t=\mu'(\bar{x}_{\tgt})\),
  as required.}

\item \emph{Completeness.}\;
  Let $\tup'$ satisfy $\phi_{\pat_{1}\concat\pat_{2}}$ in~$\db$, {
  i.e.  \(
      \mu' \in \semd{\phi_{\pat_{1}\concat\pat_{2}}}
  \).}
Then there exists $\bar u$ (the witness for~$\bar y$) such that
  \(
    \tup'_{1}\!=\!\tup'{\restriction_{\barx_{1}\cup\{\bar{x}_{\src},y\}}}
  \)
  satisfies $\phi_{\pat_{1}}$ and
  \(
    \tup'_{2}\!=\!\tup'{\restriction_{\barx_{2}\cup\{y,\bar{x}_{\tgt}\}}}
  \)
  satisfies $\phi_{\pat_{2}}$.
  By inductive hypothesis we obtain
  \(
   (\bar s,\bar u,\tup_{1})\in\sem{\pat_{1}}_{\gdb},\;
   (\bar u,\bar t,\tup_{2})\in\sem{\pat_{2}}_{\gdb}
  \)
  with $\tup_{i}\sim_{\barx_{i}}\tup'$.
  Since they agree on the common variables, their union
  $\tup:=\tup_{1}\cup\tup_{2}$ is well-defined and
  $(\bar s,\bar t,\tup)\in\sem{\pat_{1}\concat\pat_{2}}_{\gdb}$
  while $\tup\sim_{\barx}\tup'$, {where \(\bar{x}=\bar{x}_{1}\cup
  \bar{x}_{2}\)}, proving completeness.
\end{enumerate}

\item[\textit{(T5) Disjunction} $\pat_{1}+\pat_{2}$.]
\mbox{}\par\addvspace{2pt}
\begin{enumerate}
  \item \emph{Soundness.}\;
  If $(\bar s,\bar t,\tup)\in\sem{\pat_{1}+\pat_{2}}_{\gdb}$ then
  $(\bar s,\bar t,\tup)\in\sem{\pat_{i}}_{\gdb}$ for some $i\!\in\!\{1,2\}$.
  Apply inductive hypothesis to $\pat_{i}$ to get
  $\tup'\models\phi_{\pat_{i}}$.
  Because $\phi_{\pat_{i}}$ is disjunctively embedded in
  $\phi_{\pat_{1}+\pat_{2}}$, the same $\tup'$ satisfies
  $\phi_{\pat_{1}+\pat_{2}}$.

\item \emph{Completeness.}\;
  Let $\tup'\models\phi_{\pat_{1}+\pat_{2}}$.
  Then $\tup'$ satisfies either $\phi_{\pat_{1}}$ or
  $\phi_{\pat_{2}}$, assume the former (the other case is symmetric).
  Denote the free variables of the sub patterns by
  \(\bar{x}_{1}\) and \(\bar{x}_{2}\), and set                
  \(
      \bar{x}\;:=\;\bar{x}_{1}\cup\bar{x}_{2} 
  \)
  Inductive hypothesis for $\pat_{1}$ gives
  $(\bar s, \bar t,\tup)\in\sem{\pat_{1}}_{\gdb}\subseteq
                  \sem{\pat_{1}+\pat_{2}}_{\gdb}$
  with $\tup\sim_{\barx}\tup'$, yielding completeness.
\end{enumerate}

\item[\textit{(T6) Bounded repetition} $\pat^{n..m}$
      \,$(0\!\le n\!\le m{<}\infty)$.]
\mbox{}\par\addvspace{2pt}
\begin{enumerate}
\item \emph{Soundness.}\;

  Assume $(\bar s, \bar t,\tup)\in\sem{\pat^{n..m}}_{\gdb}$.
  Then there exists $r\in[n,m]$ such that
  $(\bar s, \bar t,\tup)\in\sem{\pat^{r}}_{\gdb}$.
  We proceed by induction on~$r$.

  \begin{itemize}
  \item $r=0$:\;
    The assignment $\tup'$ with $\bar{x}_{\src}=\bar{x}_{\tgt}=s$
    satisfies $\bar{x}_{\src}=\bar{x}_{\tgt}=\tup'(\bar{x}_{\src})$,
    which is exactly $\phi_{\pat^{0}}$.
  \item $r>0$:\;
    Write $r=k+1$.
    Decompose the path as
    $(\bar s, \bar u,\tup_{1})\in\sem{\pat}_{\gdb}$ and
    $(\bar u, \bar t,\tup_{2})\in\sem{\pat^{k}}_{\gdb}$ with
    $\tup=\tup_{1}\cup\tup_{2}$.
    By inductive hypothesis for $\pat$ and for $\pat^{k}$ we
    obtain $\tup'_{1}\models\phi_{\pat}$ and
    $\tup'_{2}\models\phi_{\pat^{k}}$
    that agree on their common variables.
    Their union satisfies
    $\exists y(\phi_{\pat}\land\phi_{\pat^{k}})$,
    i.e.\ $\phi_{\pat^{k+1}}$, which appears in the outer
    disjunction of $\phi_{\pat^{n..m}}$, hence $\tup'$ is a witness.
  \end{itemize}

\item \emph{Completeness.}\;
  Let $\tup'$ satisfy
  $\displaystyle
    \phi_{\pat^{n..m}}\;=\;
      \bigvee_{r=n}^{m}\phi_{\pat^{r}}$.
  Fix the least $r$ for which $\tup'\models\phi_{\pat^{r}}$
  and prove, by induction on~$r$, the existence of
  $(\bar s, \bar t,\tup)\in\sem{\pat^{r}}_{\gdb}$ with
  $\tup\sim_{\barx}\tup'$.

  The argument mirrors soundness:
  for $r=k+1$ extract $\bar u=\tup'(\bar y)$ from the witness of the existential,
  invoke inductive hypothesis on $\pat$ and on $\pat^{k}$,
  and combine the resulting compatible mappings.
  Since $r\!\in\![n,m]$, the obtained triple lies in
  $\sem{\pat^{n..m}}_{\gdb}$ as required.
\end{enumerate}

\item[\textit{(T7) Filtering} $\pat\langle\theta\rangle$.]
\mbox{}\par\addvspace{2pt}
\begin{enumerate}
\item \emph{Soundness.}\;
  If $(\bar s, \bar t,\tup)\in\sem{\pat\langle\theta\rangle}_{\gdb}$ then
  $(\bar s, \bar t,\tup)\in\sem{\pat}_{\gdb}$ and
  the valuation $\tup$ satisfies the filter
  $\theta$ inside $\gdb$.
  Inductive hypothesis gives
  $\tup'\models\phi_{\pat}$, and the truth of $\theta$ carries
  over verbatim to the FO translation $\theta^{\mathrm{FO}}$.
  Therefore $\tup'\models
          \phi_{\pat}\land\theta^{\mathrm{FO}}
          =\phi_{\pat\langle\theta\rangle}$.

\item \emph{Completeness.}\;
  Conversely, let
  $\tup'\models\phi_{\pat\langle\theta\rangle}
              =\phi_{\pat}\land\theta^{\mathrm{FO}}$.
  Then $\tup'$ satisfies $\phi_{\pat}$, so by inductive hypothesis there is $(\bar s, \bar t,\tup)\in\sem{\pat}_{\gdb}$ with
  $\tup\sim_{\barx}\tup'$.
  Because $\tup'$ makes $\theta^{\mathrm{FO}}$ true,
  the original mapping $\tup$ fulfils the semantic filter,
  hence $(\bar s, \bar t,\tup)\in\sem{\pat\langle\theta\rangle}_{\gdb}$,
  completing the argument.
\end{enumerate}

\item[\textit{(T8) Kleene star} $\pat^{*}$.]
\mbox{}\par\addvspace{2pt}
\begin{enumerate}
\item \emph{Soundness.}\;
  Assume
  \(
    (\bar s, \bar t,\tup)\in\sem{\pat^{*}}_{\gdb}.
  \)
  By the definition of $\sem{\cdot}_{\gdb}$ there exists
  an integer $n\!\ge\!0$, a sequence of nodes
$\bar u_{0}=\bar s,\dots,\bar u_{n}=\bar t$
  and mappings
  \(
    \tup_{1},\dots,\tup_{n}
  \)
  such that, for every
  $0\le i<n$,
$(\bar u_{i},\bar u_{i+1},\tup_{i+1})\in\sem{\pat}_{\gdb}$
  and
  \(
    \tup=\tup_{1}\cup\dots\cup\tup_{n}
  \)

  For each segment apply the inductive hypothesis~(a) to obtain
  assignments
  \[
\tup'_{i+1}\models\phi_{\pat}(\barx,\bar u_{i},\bar u_{i+1})
    \qquad(0\le i<n).
  \]
  Because $\mathrm{TC}$ is reflexive
  (i.e.\ $(\mathrm{TC}_{u,v}\phi)(a,a)$ always holds,
  capturing the empty, length-$0$ path),
  the existence of the chain
  \(
    u_{0},\dots,u_{n}
  \)
  witnesses
  \(
    \bigl(\mathrm{TC}_{\bar u,\bar v}\,
         \phi_{\pat}(\barx,\bar u,\bar v)\bigr)(\bar s,\bar t).
  \)
  Let
  \(
    \tup' \df \tup{\restriction_{\barx}}
        \cup\{\bar{x}_{\src}\mapsto \bar s,\,
              \bar{x}_{\tgt}\mapsto \bar t\}.
  \)
  Then
  \(
    \tup'\models\phi_{\pat^{*}}(\barx,\bar{x}_{\src},\bar{x}_{\tgt}),
  \)
  establishing soundness.

\item \emph{Completeness.}\;
  Conversely, suppose an assignment
  \(
    \tup'\models\phi_{\pat^{*}}(\barx,\bar{x}_{\src},\bar{x}_{\tgt})
  \)
  in~$\db$.
  Write
  \(
    \bar s\df\tup'(\bar{x}_{\src}),\;
    \bar t\df\tup'(\bar{x}_{\tgt}).
  \)
  By the semantics of
  \(
    \bigl(\mathrm{TC}_{\bar u,\bar v}\phi_{\pat}\bigr)(\bar s,\bar t)
  \)
  there exists $n\!\ge\!0$ and a sequence
$\bar u_{0}=\bar s,\dots,\bar u_{n}=\bar t$
  such that
  \(
    \phi_{\pat}(\barx,\bar u_{i},\bar u_{i+1})
  \)
  holds under~$\tup'$
  for all $0\le i<n$.
  Apply the inductive hypothesis to each of these $n$
  instantiations, obtaining
  triples
  \(
(\bar u_{i},\bar u_{i+1},\tup_{i+1})\in\sem{\pat}_{\gdb}
  \)
  with
  \(
    \tup_{i+1}\sim_{\barx}\tup'.
  \)
  Since all $\tup_{i+1}$ agree on $\barx$ they are pairwise
  compatible, let
  \(
    \tup\df\tup_{1}\cup\dots\cup\tup_{n}
  \).
  Then
  \(
    (\bar s, \bar t,\tup)\in\sem{\pat^{*}}_{\gdb},
  \)
  and
  \(
    \tup\sim_{\barx}\tup',
  \)
  proving completeness.
\end{enumerate}
\end{description}

All pattern constructors have now been considered- therefore
$\phi_{\pat}$ satisfies both clauses, completing the proof.
\end{proof}

\begin{reptheorem}{\ref{thm:erwpgq-in-fotc}}
    \thmerwpgqinfotc
\end{reptheorem}
\begin{proof}
We define, by structural recursion on the query syntax of
$\erwpgq$ (Figure~\ref{fig:grammar-core-pgq}), a translation
\[
    \tau\colon \erwpgq \longrightarrow \mathrm{FO[TC]},
    \qquad
    \tau(Q)=\varphi_{Q}(x_1,\dots,x_n)
\]
such that for every relational database\/
$\db$ over~$\schema$,
\begin{equation}\label{eq:tau-correct}
    \semd{Q}\;=\;\semd{\varphi_{Q}(\bar{x})},
\end{equation}
where $\bar{x}$ is the tuple of all free variables of~$Q$, explicitly denoted as $(x_1,\dots,x_n)$
(by construction\/ $|\bar{x}|=\arity{Q}$).
Equation~\eqref{eq:tau-correct} is proved by induction on~$Q$,
and the required formula of the theorem is then
$\varphi_{Q}:=\tau(Q)(x_{1},\dots,x_{n})$.

  \textbf{Base cases.}
  \begin{enumerate}
      \item\emph{Relation atom.}
            For $Q=R$ with $\arity{R}=k$ put
            \[
              \tau(R)
                \;:=\;
              R(x_{1},\dots,x_{k}).
            \]
            Because $\semd{R}=R^{\db}$,~\eqref{eq:tau-correct} is immediate.

      \item\emph{Individual constant.}
            For $Q=c$ (a $1$-arity query that returns the domain element~$c$)
            let $\tau(c)(x)\;:=\;(x{=}c)$.
            Then $\semd{c}=\{(c)\}$ and
            $\semd{(x=c)}=\{(c)\}$, so~\eqref{eq:tau-correct} holds.
  \end{enumerate}

\textbf{Induction steps.}

  \textbf{Relational operators.}
Let $Q_{1},Q_{2}$ be queries with arities
$k_{1}$ and $k_{2}$, and by induction hypothesis put
\[
    \tau(Q_{i})=\varphi_{i}\!\bigl(\bar{x}_{i}\bigr)
    \quad\text{where}\quad
    \bar{x}_{i}=(x^{i}_{1},\dots ,x^{i}_{k_{i}}).
\]

(The tuples $\bar{x}_{1}$ and $\bar{x}_{2}$ are taken pairwise
distinct.)  
Define~$\tau$ on compound queries as follows:
\[
\renewcommand{\arraystretch}{1.3}
\begin{array}{rcllll}
Q_{1}\cup Q_{2} &:&
  \tau(Q)(\bar{x}_{1}) &:=&
  \varphi_{1}(\bar{x}_{1})\;\vee\;
  \varphi_{2}(\bar{x}_{1})
  & (k_{1}=k_{2},\;\bar{x}_{1}=\bar{x}_{2})\\[2pt]    
Q_{1}\setminus Q_{2} &:&
    \tau(Q)(\bar{x}_{1}) &:=&
    \varphi_{1}(\bar{x}_{1})\;\wedge\;
    \neg\varphi_{2}(\bar{x}_{1})
    & (k_{1}=k_{2})\\[2pt]
Q_{1}\times Q_{2} &:&
    \tau(Q)(\bar{x}_{1},\bar{x}_{2}) &:=&
    \varphi_{1}(\bar{x}_{1})\;\wedge\;
    \varphi_{2}(\bar{x}_{2})\\[2pt]
\pi_{i_{1},\dots ,i_{m}}(Q_{1}) &:&
    \tau(Q)(x_{1},\dots ,x_{m}) &:=&
    \exists \bar{z}\;
      \bigl(
        \varphi_{1}(\bar{z})\;\wedge\;
        \textstyle\bigwedge_{j=1}^{m} x_{j}=z_{i_{j}}
      \bigr)\\[2pt]
\sigma_{\theta}(Q_{1}) &:&
    \tau(Q)(\bar{x}_{1}) &:=&
    \varphi_{1}(\bar{x}_{1})\;\wedge\;\theta(\bar{x}_{1}).
\end{array}
\]
For selections we fix a first order (quantifier free)
formula $\theta$ whose only free variables occur in the tuple
$\bar{x}$, we denote this dependence by $\theta(\bar{x})$.

  In each row the semantic equivalence~\eqref{eq:tau-correct}
  follows directly from the semantics of relational algebra.

  \textbf{Inductive \erwpgq-specific case.}
  Let
  $Q=\patExtended{\return}(\bar{Q})$
  with
  $\bar{Q}=(Q_{1},\dots,Q_{6})$.
  By the outer induction hypothesis we have already constructed
  translations
  $\tau(Q_{i})=\varphi_{i}$ for every component.
  Denote
  \(
      G\;=\;
      \pgViewExtendedUnion \!\bigl(
        \semd{Q_{1}},\dots,\semd{Q_{6}}
      \bigr).
  \)
  Applying Lemma~\ref{lem:pattern-translation} to the pattern
  $\pat$ and to the $6$-tuple
  $(\varphi_{1},\dots,\varphi_{6})$
  yields an $\mathrm{FO[TC]}$ formula
  \(
      \varphi_{\pat}
      (\,\bar{y},\bar{x}_{\src},\bar{x}_{\tgt}\,)
  \)
  whose free variable tuple is
  $\{\bar{x}_{\src}\}\cup\bar{y}\cup\{\bar{x}_{\tgt}\}$
  and that satisfies the soundness and completeness
  clauses of the lemma with respect to~$G$.

  Let\/ $\bar{z}=\sch{\patExtended{\return}}$
  be the variables mentioned in $\return$.
  Define
  \[
     \tau\!\bigl(\,\patExtended{\return}(\bar{Q})\,\bigr)
     (\bar{z})
     \;:=\;
     \exists \bar{x}_{\src}\,\bar{x}_{\tgt}
     \; \varphi_{\pat}\bigl(\bar{z},\bar{x}_{\src},\bar{x}_{\tgt}\bigr).
  \]
  For any assignment~$\beta$ to~$\bar{z}$ we have
\begin{align*}
  \db\models\tau(Q)[\beta]
      &\;\Longleftrightarrow\;
        \text{there exist }a,b
        \text{ with }
        G\models\varphi_{\pat}[\beta\cup\{\bar{x}_{\src}\!\mapsto\!a,
                                        \bar{x}_{\tgt}\!\mapsto\!b\}] \\
      &\;\Longleftrightarrow\;
        \beta(\bar{z})\in
        \sem{\patExtended{\return}(\bar{Q})}_{\db},
\end{align*}
establishing~\eqref{eq:tau-correct} for the pattern case.

  The translation~$\tau$ is defined for every syntactic construct,
  and each clause preserves semantics
  by either direct calculation (relational operators)
  or by the lemma (for the pattern case).
  Hence property~\eqref{eq:tau-correct} is established for
  all\/ $Q\in\erwpgq$.
  Set $\varphi_{Q}:=\tau(Q)$.
  Then~\eqref{eq:tau-correct} gives
  $\semd{Q}=\semd{\varphi_{Q}(\bar{x})}$,
  proving the inclusion
  $\erwpgq\subseteq\mathrm{FO[TC]}$ as claimed.
\end{proof}

\begin{lemma}\label{lem:tc-translation}
    \lemtctranslation
\end{lemma}

\begin{proof}
We prove by structural induction on the $\mathrm{FO[TC]}$ formula
$\psi$.  The induction hypothesis gives, for every proper sub-formula,
an \erwpgq~query with equivalent semantics.  The only non-first-order
construction is {transitive closure, so it suffices to treat the
case
\[
  \psi(\bar x,\bar y,\bar p)
  \;=\;
  \mathrm{TC}_{(\bar u,\bar v)}
    \bigl[\varphi(\bar u,\bar v,\bar p)\bigr]
    (\bar x,\bar y),
\]
By the induction
hypothesis there exists an \erwpgq~query
\(Q_{\varphi}\) such that
\(
  \sem{Q_{\varphi}}_{\db}=\sem{\varphi}_{\db}
\)
for every database~$\db$.

Let
\[
  \mathit{C}
  \;:=\;
  \pi_{\bar p}\bigl(Q_{\varphi}\bigr)
  \qquad(\subseteq\ADOM^{\lvert\bar p\rvert}).
\]
Intuitively, \(\mathit{C}\) enumerates every parameter tuple
\(\bar c\) for which \(\varphi\) is potentially satisfiable.

For each concrete parameter tuple
\(\bar c\in \mathit{C}\) define the \erwpgq~sub-query
\[
  E_{\bar c}
  \;:=\;
  \pi_{\bar u,\bar v}\!
    \bigl(
      \sigma_{\bar p=\bar c}(Q_{\varphi})
    \bigr).
\]
By construction,
\(
  \sem{E_{\bar c}}_{\db}
  =\{
     (\bar a,\bar b)
     \mid
     \db\models
     \varphi(\bar a,\bar b,\bar c)
   \}
\).
Write
{
\(
  N_{\bar c}
  \;:=\;
  \pi_{\bar u,\bar u}(E_{\bar c})
  \;\cup\;
  \pi_{\bar v,\bar v}(E_{\bar c})
\) for its node set.
The duplication of each $k$-tuple ensures that
$N_{\bar c}$ has arity $2k$, exactly matching the arity of
$E_{\bar c}$.  We impose this merely because
\(\pgView^{\mathrm{ext}}\) (Definition~\ref{def:pgView-n}) requires these relations to share the same arity; as pointed out in
Remark~\ref{rmk:arity-remark}, the rest of the argument would go
through even without this syntactic uniformity.}

Using the extended graph constructor from Section~\ref{sec:extensions},
\[
  G_{\bar c}
  \;:=\;
  \pgViewExtendedUnion\!
     \bigl(
       N_{\bar c},\,
       E_{\bar c},\,
       \pi_{\bar v}(E_{\bar c}),\,
       \pi_{\bar u}(E_{\bar c}),\,
       \varnothing,\,\varnothing
     \bigr).
\]
Thus \(G_{\bar c}\) has
\(\lvert\bar u\rvert=\lvert\bar v\rvert=k\)-tuples as
nodes.

Denote by
\[
  \pat_{\mathit{reach}}
  \;:=\;
  (\bar x)\;\xrightarrow{}^{\!\!*}\;(\bar y)
\]
the $k$-tuple variant of the standard
two-node Kleene-star pattern.
When evaluated on \(G_{\bar c}\) it returns the relation
\(
  \{
    (\bar a,\bar b)
    \mid
    \text{a path from }\bar a\text{ to }\bar b\text{ exists in }G_{\bar c}
  \}
\).

\erwpgq{} allows constants to be injected via selections, hence the
Cartesian product in
\[
  Q_{\varphi^{\mathrm{TC}}}
  \;:=\;
  \bigcup_{\bar c\in\mathit{C}}
    \Bigl(
      \pat_{\mathit{reach}}\!\bigl(G_{\bar c}\bigr)
      \;\times\;
      \{\bar c\}
    \Bigr)
\]
is realized by the ordinary join
\(
  \pat_{\mathit{reach}}(G_{\bar c})
  \Join
  \sigma_{\bar p=\bar c}(\mathit{C}).
\)
Because \erwpgq~supports union,
selection, projection and joins, the query
\(Q_{\varphi^{\mathrm{TC}}}\) is again an \erwpgq~expression whose
schema is exactly \(\bar x\cup\bar y\cup\bar p\).

Now, we prove the Correcntess: Fix a database~$\db$ and a valuation
\(\mu\colon\bar x\cup\bar y\cup\bar p\to\ADOM\).
Write \(\bar a:=\mu(\bar x)\), \(\bar b:=\mu(\bar y)\),
and \(\bar c:=\mu(\bar p)\).

\begin{itemize}
  \item[($\Rightarrow$)]
        Assume
        \(
          \mu\in
          \sem{\varphi^{\mathrm{TC}}}_{\db}
        \).
        Then by the semantics of transitive closure, there exists a
        path
        \(
          \bar a=\bar t_{0},
          \bar t_{1},\dots,
          \bar t_{m}=\bar b
        \)
        such that
        \(
          \db\models
          \varphi(\bar t_{i},\bar t_{i+1},\bar c)
        \)
        for every $0\le i<m$.
        Hence $(\bar t_{i},\bar t_{i+1})\in\sem{E_{\bar c}}_{\db}$
        and the same sequence is a path in
        $G_{\bar c}$, showing
        \(
          (\bar a,\bar b)\in
          \sem{\pat_{\mathit{reach}}(G_{\bar c})}_{\db}
        \).
        Because $\bar c\in\sem{\mathit{C}}_{\db}$,
        the join contributes
        $(\bar a,\bar b,\bar c)$ to
        $\sem{Q_{\varphi^{\mathrm{TC}}}}_{\db}$,
        i.e.\ \(\mu\in\sem{Q_{\varphi^{\mathrm{TC}}}}_{\db}\).

  \item[($\Leftarrow$)]
        Conversely, suppose
        \(
          \mu\in
          \sem{Q_{\varphi^{\mathrm{TC}}}}_{\db}
        \).
        Then there is a parameter tuple
        \(\bar c\in\sem{\mathit{C}}_{\db}\)
        such that
        \(
          (\bar a,\bar b)\in
          \sem{\pat_{\mathit{reach}}(G_{\bar c})}_{\db}
        \).
        Therefore a path
        \(
          \bar a=\bar t_{0},
          \bar t_{1},\dots,
          \bar t_{m}=\bar b
        \)
        exists in $G_{\bar c}$, meaning
        \(
          (\bar t_{i},\bar t_{i+1})
          \in
          \sem{E_{\bar c}}_{\db}
        \)
        and thus
        \(
          \db\models
          \varphi(\bar t_{i},\bar t_{i+1},\bar c)
        \)
        for every edge of the path.
        By the definition of $\mathrm{TC}$,
        \(
          \db\models
          \varphi^{\mathrm{TC}}(\bar a,\bar b,\bar c)
        \),
        i.e.\ \(\mu\in\sem{\varphi^{\mathrm{TC}}}_{\db}\).
\end{itemize}

\noindent
Since the two inclusions hold for every $\db$ and every valuation,
\(
  \sem{Q_{\varphi^{\mathrm{TC}}}}_{\db}
  =\sem{\varphi^{\mathrm{TC}}}_{\db}
\)
as required.}
\end{proof}

\begin{reptheorem}{\ref{thm:fotc-contained-in-erwpgq}}
    \thmfotcinrerwpgq
\end{reptheorem}

\begin{proof}
  We give a translation
  \(
    T : \mathrm{FO[TC]}\longrightarrow\erwpgq
  \)
  and prove by structural induction on every formula~$\varphi$ that
  \(
    \semd{T(\varphi)}=\semd{\varphi}
  \).
  Throughout the proof, let
  \(
    Q_{A}:=\bigcup_{R\in\schema}\;
             \bigcup_{1\le i\le\arity{R}}\pi_{i}(R)
  \)
  be the \emph{active domain} query, it is an
  \erwpgq~query because each $R$ is and due to finiteness of $\schema$.  We also abbreviate
  \(
    A^{(k)}:=\underbrace{Q_{A}\times\cdots\times Q_{A}}_{k\ \text{factors}}
  \).
  
  \textbf{Base cases.}

  \textbf{Atomic predicates.}
  \begin{enumerate}
    \item[] {\em Relation atom $R(\bar x)$.}  
           Put $T(R(\bar x)):=R$.
           Because the database already stores~$R$, this query returns
           exactly $\semd{R(\bar x)}$.

    \item[] {\em Equality $x=y$.}  
           Define
           \(
             T(x=y):=\sigma_{\$1=\$2}\!~(Q_{A}\times Q_{A}).
           \)
           Both projections coincide with the intended binary equality
           relation.
  \end{enumerate}

    \textbf{Induction steps.}  
    
    \textbf{Boolean connectives.}
  Assume, using induction hypothesis, $T(\varphi_1)$ and $T(\varphi_2)$ are already defined
  and have identical schemas $\bar x$ of length~$k$. 
  \[
    \begin{aligned}
      T(\varphi_1\;\wedge\;\varphi_2) &:= T(\varphi_1)\cap T(\varphi_2),\\
      T(\varphi_1\;\vee\;\varphi_2)   &:= T(\varphi_1)\cup T(\varphi_2),\\
      T(\neg\varphi_1)               &:= A^{(k)}\setminus T(\varphi_1).
    \end{aligned}
  \]
  \textbf{Quantifiers.}
  For a formula $\varphi(\bar x,y)$ with translation
  $T(\varphi)$ from the induction hypothesis, we set
  \[
    T(\exists y\,\varphi) := \pi_{1,\ldots,|\bar x|}\!\bigl(T(\varphi)\bigr),
    \qquad
    T(\forall y\,\varphi):=
      A^{(|\bar x|)}\setminus
      \pi_{1,\ldots,|\bar x|}\!\bigl(
        A^{(|\bar x|+1)}\setminus T(\varphi)
      \bigr).
  \]

  \textbf{Transitive closure.}
  Suppose
  \(
    \varphi(\bar u,\bar v,\bar p)
  \)
  is an {
  $\mathrm{FO[TC]}$} formula whose transitive closure
  \[
    \psi=\mathrm{TC}_{\bar u,\bar v}[\varphi](\bar x,\bar y)
  \]
  occurs inside a larger expression.
  Apply Lemma~\ref{lem:tc-translation} to obtain an
  \erwpgq~query
  \(
    Q_{\psi}
  \)
  with
  \(
    \semd{Q_{\psi}}=\semd{\psi}
  \).  Define
  \(
    T(\psi):=Q_{\psi}.
  \)
  All free variables of~$\psi$ appear unchanged in
  $Q_{\psi}$, so subsequent Boolean operations and
  quantifiers are well-typed.

  By construction, each clause of $T$ produces an
  \erwpgq~query whose schema coincides with the free-variable tuple
  of the input formula.  The base and inductive cases above show
  equality of semantics for every syntactic constructor, therefore
  the claim
  \(
    \semd{T(\varphi)}=\semd{\varphi}
  \)
  holds for all $\mathrm{FO[TC]}$ formulas~$\varphi$.

  For every $\mathrm{FO[TC]}$ formula $\varphi(\bar x)$ we have produced
  an \erwpgq~query $Q_\varphi:=T(\varphi)$ with the same free-variable
  tuple~$\bar x$ and identical evaluation on every relational
  database.  Hence
  \(
    \mathrm{FO[TC]}\subseteq\erwpgq,
  \)
  completing the proof.
\end{proof}

\begin{reptheorem}{\ref{thm:pgqextn-in-fotcn}}
    \thmpgqextninfotcn
\end{reptheorem}
\begin{proof}
Fix \(n\ge 1\).
Define \(\tau_{n}\colon\pgqextn\to\fotcn\) exactly as in the proof of
Theorem~\ref{thm:erwpgq-in-fotc} except for the pattern-translation
clause.

The translations for relation atoms, constants, and \(\{\cup,\setminus,
\times,\pi,\sigma\}\) are exactly the same as Theorem~\ref{thm:erwpgq-in-fotc}, their correctness arguments
carry over unchanged because those are set operation on tuples that
does not reference node or edge identifiers.
Hence \(\semd{Q}\) is the same whether identifiers are
unary or \(n\)-ary, and the original correctness proof
\(\semd{Q}=\semd{\tau(Q)}\)
carries over verbatim,
giving
\(
  \semd{Q}=\semd{\tau_{n}(Q)}.
\)
Only the pattern clause depends on identifier arity and is
adjusted in the proof below.

Let
\[
   Q=\pat^{n}_{\return}(Q_{1},\dots ,Q_{6})
\]
By induction we already have
\(
  \tau_{n}(Q_{i})=\varphi_{i} { 
  \in \fotcn}
\)
satisfying
\(
  \semd{Q_{i}}=\semd{\varphi_{i}}.
\)

Set
\(
  G_{n}:=\pgView^{n}\!\bigl(\semd{Q_{1}},\dots ,\semd{Q_{6}}\bigr)
\)
(Definition~\ref{subsec:inside-erwpgq}), where identifiers are \(n\)-tuples.

Introduce fresh \(n\)-tuples
\(
  \bar{x}_{\src},\bar{x}_{\tgt}
\).

{

Applying Lemma~\ref{lem:pattern-translation} to the pattern \(\pat\) involves explicitly translating constructs such as the Kleene star pattern \(\pat^{*}\) into a transitive closure operator as follows:
\[
  \tau(\pat^{*})(\bar{x},\bar{x}_{\src},\bar{x}_{\tgt})\;\df\;
   \bigl(
     \mathrm{TC}_{\bar{u},\bar{v}}\;
       \tau(\pat)(\bar{x},\bar{u},\bar{v})
   \bigr)\!(\bar{x}_{\src},\bar{x}_{\tgt})
\]

Here, the key reason the transitive closure operator is \(n\)-ary is precisely because the node identifiers (like \(\bar{u},\bar{v},\bar{x}_{\src},\bar{x}_{\tgt}\)) are themselves \(n\)-tuples. Thus, the transitive closure operator, which encodes reachability on these node tuples, inherently must operate over \(n\)-dimensional tuple variables. As a result, the Lemma explicitly yields a formula:
\[
   \varphi^{(n)}_{\pat}(\bar{y}, \bar{x}_{\src},\bar{x}_{\tgt})\in\fotcn,
\]
where each reachability condition is expressed through an \(n\)-ary transitive closure operator \(\mathrm{TC}^{(n)}\).}
Let \(\bar{z}=\sch{\pat^{n}_{\Omega}}\) be the variables occurring in
\(\Omega\) and define
\[
  \tau_{n}\!\bigl(\pat^{n}_{\Omega}(\bar{Q})\bigr)(\bar{z})
    :=\;
  \exists\bar{x}_{\src}\,\bar{x}_{\tgt}\;
     \varphi^{(n)}_{\pat}
       (\bar{z},\bar{x}_{\src},\bar{x}_{\tgt}).
\tag{1}
\]
By Lemma~\ref{lem:pattern-translation} with \(G_{n}\),
(1) is satisfied exactly by the tuples produced by
\(\pat^{n}_{\Omega}(\bar{Q})\), hence the induction claim holds.

Thus \(\tau_{n}(Q)(x_{1},\dots ,x_{k})\) is the desired
\(\fotcn\) formula, completing the proof.
\end{proof}

\begin{reptheorem}{\ref{thm:fotcn-in-pgqextn}}
    \thmfotcninpgqextn
\end{reptheorem}

\begin{proof}
Fix \(n\!\ge\!1\). 
We give a translation
\[
    T_{n}\colon \fotcn \longrightarrow \pgqextn,
    \qquad
    T_{n}(\varphi)=:Q_{\varphi},
\]
and prove by structural induction on every formula \(\varphi\)
that
\begin{equation}\label{eq:sem-eq}
    \semd{T_{n}(\varphi)}=\semd{\varphi}.
\end{equation}
Except for the TC-clause (treated separately below) the
definition of \(T_{n}\) coincides exactly with
the translation \(T\) used in the proof of
Theorem~\ref{thm:erwpgq-in-fotc}- we only recast its target
fragment from \(\erwpgq\) to \(\pgqextn\).

\textbf{Base cases and Boolean / quantifier clauses.}
\begin{itemize}
  \item
        \emph{Relation atom \(R(\bar{x})\) and equality \(x{=}y\).}
        Define \(T_{n}\) exactly as in the proof of
        Theorem~\ref{thm:erwpgq-in-fotc}.  The resulting queries lie in
        \(\pgqextn\) because every relation symbol is allowed as a
        primitive query.
  \item
        \emph{Boolean connectives and first-order quantifiers.}
        We use the same set algebra and projection constructions as in
        Theorem~\ref{thm:erwpgq-in-fotc}.  These operators belong to the
        core of \(\pgqextn\), equation~\eqref{eq:sem-eq} thus follows
        from the induction hypothesis.
\end{itemize}

\textbf{\(n\)-ary transitive closure clause.}
Let
\[
    \psi(\bar x,\bar y)
      \;:=\;
    \mathrm{TC}^{\,n}_{\bar u,\bar v}[\;\varphi(\bar u,\bar v,\bar p)\;]
            (\bar x,\bar y)
\]
be a sub-formula of the input.
By the outer induction hypothesis we already have
\(T_{n}(\varphi)\).
Apply Lemma~\ref{lem:tc-translation} to obtain
a \(\pgqextn\) query
\(
   Q_{\psi}
\)
whose free-variable tuple is
\(\bar x\bar y\bar p\) and whose semantics equals
\(\semd{\psi}\).
Define
\(
   T_{n}(\psi) := Q_{\psi}.
\)
Because \(\pgqextn\) admits composite identifiers and nested pattern
calls, all variables remain free and well-typed inside later Boolean
combinations or quantifier applications.

All clauses other than the TC-clause reuse verbatim the correctness
arguments of Theorem~\ref{thm:erwpgq-in-fotc}.
For the TC-clause equation~\eqref{eq:sem-eq} is guaranteed by
Lemma~\ref{lem:tc-translation}.  Hence~\eqref{eq:sem-eq} holds for every syntactic
construct, proving it for all \(\fotcn\) formulas.

For a given \(\fotcn\) formula \(\varphi(\bar{x})\) set
\(Q_{\varphi}:=T_{n}(\varphi)\).
Then \(Q_{\varphi}\in\pgqextn\), has
\(\arity{Q_{\varphi}}=|\bar{x}|\), and satisfies
\(
   \semd{\varphi(\bar{x})}=\semd{Q_{\varphi}}
\)
on every database, establishing
\(
   \fotcn \subseteq \pgqextn.
\)
\end{proof}

\end{document}